\documentclass[ASNA,twocolumn]{USG} %twocolumn
\usepackage{anyfontsize} %

\usepackage{colortbl}   % For coloring table lines
\usepackage{xcolor}     % For color definitions

\usepackage{steinmetz}

\usepackage{amsmath}
\usepackage{comment}
\usepackage{soul}

\articletype{ORIGINAL ARTICLE}%
\received{Added at production}
\revised{Added at production}
\accepted{Added at production}
\journal{Advanced Control for Applications}
\volume{00}
\copyyear{2026}
\startpage{1}
\articledoi{}

\begin{document}
\title{Cost of Sensing in Optimal Control: Basic Formulations, Examples, and Applications}
%\subtitle{Subtitle}

%\transtitle{Cost of Sensing in Optimal Control: Basic Formulations and Examples}
%\subtranstitle{trans-subtitle}

%\author[1]{Anna C. Belkina}[https://orcid.org/0000-0001-7037-2721][\facebook{https://www.facebook.com} \linkedin{https://www.linkedin.com} \twitter{https://www.twitter.com}]
%\author[2]{Caroline E. Roe}
%\author[3]{Vera A. Tang}[https://orcid.org/0000-0001-7037-2721]
%\author[4]{Jessica B. Back}[https://orcid.org/0000-0001-7037-2721]
%\author[2]{Claudia Bispo}[https://orcid.org/0000-0001-7037-2721]
%\author[6]{Alexis Conway}[https://orcid.org/0000-0001-7037-2721]

\author{Dung Tran}
\author{Tri Ngo}
\author{Tuhin Das}

\authormark{Tran \textsc{et al.}}
\titlemark{Cost of Sensing in Optimal Control: Basic Formulations and Examples}

\address{\orgdiv{Department of Mechanical and Aerospace Engineering, }\orgname{University of Central Florida, }\orgaddress{Orlando, \state{Florida, }\country{USA}}}

%\address[1]{\orgdiv{Key Laboratory for Ecological Metallurgy of Multimetallic Mineral of Ministry of Education, }\orgname{Institution Name, }%
%\orgaddress{\state{State Name, }\country{Country Name}}}

%\address[2]{\orgdiv{Department Name, }\orgname{Institution Name, }%
%\orgaddress{\state{State Name, }\country{Country Name}}}

%\address[3]{\orgdiv{Department Name, }\orgname{Institution Name, }%
%\orgaddress{\state{State Name, }\country{Country Name}}}

%\address[4]{\orgdiv{Department Name, }\orgname{Institution Name, }%
%\orgaddress{\state{State Name, }\country{Country Name}}}

\corres{Dung Tran  (\email{Dung.Tran@ucf.edu}) ~|~ 
Tuhin Das (\email{Tuhin.Das@ucf.edu})}

\editor{\textbf{Academic Editor:} ... ~|~ \textbf{Guest Editor:} ...}

%\presentaddress{This is sample for present address text this is sample for present address text.}

%\fundingInfo{National Key Research and Development Program of China, Grant/Award Number: 2021YFB1715500; National Natural Science Foundation of China, Grant/Award Number: 12072071; Scientific Research Foundation of Hunan Provincial Education Department, Grant/Award Number: 22A0104; Fundamental Research Funds for the Central Universities, Grant/Award Number: N2225027.}


%\transkeywords{Sensing-Cost, Optimal Control, Two-Point Boundary Value Problem (TPBVP), Pontryagin's Minimum Principle, Shrinking Horizon}

\abstract[ABSTRACT]{Incorporating a notion of cost of sensing, or sensing-cost, within the optimal control framework is beneficial in controlling systems where the duration of sensing, and/or the cost of sensors themselves, have a considerable impact on the overall cost. In this regard, this paper presents multiple methods for incorporating an integral sensing-cost into the optimal control framework for Linear Time-Invariant (LTI) systems. Sensing-cost is traded off against the conventional costs of control and stabilization. Optimal sensing intervals are derived by applying the Pontryagin's Minimum Principle. Other formulations of the sensing-cost problem, and extension to nonlinear systems, are possible. The theoretical developments of this paper are validated through numerical solutions and demonstrated through simulations. A reduced-form expression for the infinite-horizon multi-dimensional case with single switching point is derived, and a closed-form solution is obtained for the infinite-horizon first-order case. Additionally, a Shrinking Horizon method is demonstrated for practical implementation of the proposed theory and as a means to address uncertainties. A practical case study of a wastewater treatment plant is introduced to examine the applicability of sensing-cost considerations in a real-world setting.}

%\transabstract[transABSTRACT]{This is a generic template designed for use by multiple journals, which includes several options for customization. Please refer the author guidelines and author LaTeX manuscript preparation document for the journal to which you are submitting in order to confirm that your manuscript will comply with the journal’s requirements. Please replace this text with your abstract. This is sample abstract text just for the template display purpose.}

%\abbr{5-FU, 5-fluorouracil; CFD, computational fluid dynamics; CH, channel; EFS, event-free survival; GBM, glioblastoma multiforme; OS, overall survival; PFS, progression-free survival; SD, standard deviation.}

%\contributed{}

%\dedicated{}

%\copyright{This is an open access article under the terms of the \href{Creative Commons Attribution-NonCommercial}{Creative Commons Attribution-NonCommercial} License, which permits use, distribution and reproduction in any medium, provided the original~work~is~properly cited and is not used for commercial purposes.
%\\[5pt]
%©  2026 The Author(s) \textit{Advanced Control for Applications} published by Wiley Periodicals LLC on behalf of American Institute of Chemical Engineers.}

%\openaccessstatement

\maketitle
%\afterpage{\aftergroup\restoregeometry}

%\onecolumn

%%%%%%%%%%%%%%%%%%%%%%%%%%%%%%%%%%%%%%%%%%%%%%%%%%%%%%%%%%%%%%%%%%%%%%%%%%%%%%%%%%%%%%%%%%%%%%%%%%%%%%%
%%%%%%%%%%%%%%%%%%%%%  End of fields to be completed. Now write! %%%%%%%%%%%%%%%%%%%%%%%%%%%%%%%%%%%%%%

\section{Introduction}
%In control theory and its applications, the cost of sensing plays a pivotal role in shaping the design, implementation, and performance of control systems. Sensing cost encompasses various aspects, including financial expenses, computational overhead, energy consumption, and impacts on system complexity and reliability. This paper aims to explore the implications of sensing cost in control theory, summarizing key methodologies, findings, and trends in this domain.

In conventional optimal control problems, such as in LQR, the performance index typically weighs transient performance vs. control effort. An additional consideration can be the cost of sensing. Cost of sensing, or alternately sensing-cost, plays a pivotal role in shaping the overall design, implementation, and performance of control systems. It encompasses various aspects, including the price of sensors, the computational overhead associated with sensing and processing, energy consumption, and impacts on system complexity and reliability. This paper aims to explore the implications of incorporating sensing-cost within the optimal control framework, summarizing basic formulations, key methodologies, and findings.

Incorporation of sensing-cost within optimal control has not been reported in the literature. However, due to their ubiquitousness and practicality in control systems, numerous works have addressed various facets of sensing, such as sensor placement, scheduling, sampling, etc. While such works are not directly relevant to this paper, we discuss selected ones from literature to put this work in context. For instance, \cite{clark2019greedy} propose optimization techniques for sensor placement, employing greedy algorithms to achieve cost-effective configurations that maintain desired performance. In \cite{das2008shared}, the authors propose the concept of shared-sensing for reversible transducers that are continuously switched between actuator and sensor modes. The authors in \cite{ganesan2017vibration} introduce a method to reduce the cost and power requirements of sensing while streamlining data storage and processing in vibration-based monitoring and diagnostics using compressive sensing. In \cite{gupta2006stochastic}, the authors address collaborative sensing with the goal of determining a sensor schedule that minimizes the error covariance. In \cite{savkin2001problem}, the authors address a sensor scheduling problem that involves estimating the state of an uncertain process using measurements obtained by sequentially switching among a set of noisy sensors.

%Sensors are fundamental to control systems, providing critical information for tasks such as state estimation, feedback control, and fault detection. Accurate and efficient sensing ensures precise knowledge of the system’s state, enabling robust control actions. However, the placement, operation, and management of sensors introduce costs that must be optimized without compromising system performance. The selection and placement of sensors are crucial factors affecting cost and performance. For instance, Clark et al. \cite{clark2018greedy} propose optimization techniques for sensor placement, employing greedy algorithms to achieve cost-effective configurations that maintain desired performance.

%Advancements in computer and communication technologies have enabled the development of a new class of large-scale, resource-constrained wireless embedded control systems. In these systems, it is preferable to limit sensor measurements, control computations, and communications to instances when the system requires attention. However, traditional sampled-data control relies on periodic sensing and actuation, regardless of whether intervention is needed. 

A work which has some conceptual similarity with that presented in this paper is \cite{heemels2012introduction}. Here, the authors introduce event-triggered and self-triggered control systems. Event-triggered control is reactive, initiating sensor sampling and control actions when the system state deviates beyond a predefined threshold. In contrast, self-triggered control takes a proactive approach, predicting the next sampling or actuation instance in advance. 
% The fundamentals of these control strategies are presented, along with a discussion on the distinctions between state feedback and output feedback in event-triggered control. 
The paper demonstrates how these control techniques can be implemented, with examples from applications in the process industry. In this work, the states are sensed regularly to detect triggering, while action (i.e., control computation and new actuation) may only take place when triggered.

In many real-world systems, the challenge lies not in sensor availability but in the cost of continuous sensing—whether in energy consumption, bandwidth usage, signal processing/computational cost or operational constraints. Battery-powered UAVs, underwater autonomous vehicles, wireless sensor networks, and planetary rovers all operate under strict power budgets, where high-frequency sensing (e.g., LiDAR, sonar) can significantly shorten mission duration. Similarly, in networked control systems—often with battery-operated sensors deployed in remote locations—communication costs further constrain sensing schedules, forcing control designers to explicitly consider when sensors should be activated. For example, in \cite{boel2015optimal}, the authors quantify the trade-off between the cost of activating expensive sensors and the resulting improvement in observation accuracy. In such cases, the control agent’s task extends beyond stabilizing system states to also accounting for sensing-costs, with the aim of minimizing the overall operational costs. Optimizing sensing intervals in these contexts is therefore critical to extending mission longevity and maintaining system performance.

However, there is no rigorous study that directly incorporates sensing-cost into the optimal control problem, as proposed in this paper. This approach jointly accounts for sensing-cost, control effort, and state stabilization, ensuring a trade-off between system performance and resource utilization. Here, we present the fundamental formulations of sensing-cost, along with simple examples to illustrate the theoretical development. The contributions of this paper are:
\begin{enumerate}
    \item 
We propose a unified optimization framework that explicitly includes sensing-cost alongside control effort and state stabilization in the cost functional. This framework enables the systematic evaluation of the trade-offs between these competing objectives, addressing the limitations of traditional control design approaches that treat cost of sensing as a secondary concern or ignore it altogether.
\item
A binary switching variable is used to toggle state sensing on and off. Switching is optimized to minimize the cost functional, determining the optimal sequence of sensing actions based on system dynamics and desired performance criteria. We employ the Pontryagin's Minimum Principle to derive the necessary conditions for optimality.
\item
The proposed framework is validated through numerical simulations, demonstrating its effectiveness in balancing sensing-cost and control performance. A Shrinking Horizon approach is demonstrated, providing a practical alternative to implementing the resulting Two-Point Boundary Value Problems (TPBVPs).

\end{enumerate}

The remainder of the paper is organized as follows. Section \ref{sec_ps} defines the problem, introducing the system model, cost functional, and constraints. In Section \ref{sect:theory}, the theoretical development of the proposed framework is detailed, including the derivation of optimality conditions using Pontryagin's Minimum Principle. Section \ref{sec_1sw} examines a specific case of the sensing cost framework for which an analytical expression can be derived for the infinite-horizon optimal problem. This section also presents a closed-form solution for the first-order infinite-horizon case. Numerical examples showcasing the effectiveness of the approach across various scenarios are presented in Section \ref{section:Examples}. To address practical implementation challenges, Section~\ref{section:WWTP} presents an example of incorporating sensing cost into a wastewater treatment plant and Section \ref{sec_sh} introduces a Shrinking Horizon method for dynamic optimization under uncertainties. Finally, Section \ref{sec_conclu} provides concluding remarks and outlines directions for future research.

%%%%%%%%%%%%%%%%%%%%%%%%%%%%%%%%%%%%%%%%%%%%%%%%%%%%%%%%%%%%%%%%
\section{Problem Statement}
\label{sec_ps}
Consider a Linear Time Invariant (LTI) system,
\begin{equation}
\label{eqn:sys}
    \dot{X} = AX + Bu, \quad X \in \mathbb{R}^n, \quad u \in \mathbb{R}^m
\end{equation}
Consider the design of a controller for this system under the condition that sensing of the states $X$ can be switched on or off. We define a scalar switching variable $z$ as follows,
\begin{equation}
    z = \left\{ 
    \begin{array}{ll}
    1 & \quad \mbox{when sensing is on} \\
    0 & \quad \mbox{when sensing is off}
    \end{array}
    \right.
\end{equation}
When $z = 1$, we consider two types of input formulations,
\begin{equation}
    z = 1 \; \Rightarrow \; \left\{ 
    \begin{array}{l}
    \mbox{Case 1: Predefined state feedback $u = -KX$} \\
    \mbox{Case 2: Optimal formulation $u^*$}
    \end{array}
    \right.
    \label{eq_z1cases}
\end{equation}

%\hl{
%During sensing-off intervals, we do not propagate a state estimate. Although model-based propagation is possible under perfect model knowledge and initialization, in practice it becomes an open-loop prediction that is sensitive to modeling errors, especially in the shrinking-horizon setting of Section~}
%\ref{sec_sh}. 
%\hl{Since the sensing-off interval is intended to represent a lack of information rather than state estimation, a constant control input is applied instead:}
Similarly, when $z = 0$ we consider the following two types of input formulations,

\begin{equation}
    z = 0 \; \Rightarrow \; \left\{ 
    \begin{array}{l}
    \mbox{Case 1:} \quad \dot{u} = \left[0\right] \\
    \mbox{Case 2:} \quad u = \left[0\right]
    \end{array}
    \right.
    \label{eq_z0cases}
\end{equation}
We impose that $z$ must be switched to minimize the cost functional,
\begin{equation}
    J = \int_{t_0}^{t_f} L\left( X, u, z \right) dt
    %J = \int_{t_0}^{t_f} \left( X^T Q X + u^T R u + sz^2 \right) dt
    \label{eq_costf1}
\end{equation}
where $t_f$ is a fixed final time. Based on the above problem statement, in the next section we formulate three specific optimal control problems, with examples provided in subsequent sections. Each of the three problems incorporates the cost of sensing and provides a switching law for $z$.

%%%%%%%%%%%%%%%%%%%%%%%%%%%%%%%%%%%%%%%%%%%%%%%%%%%%%%%%%%%%%%%%%
\section{Theoretical Development}
\label{sect:theory}

%%%%%%%%%%%%%%%%%%%%%%%%%%%%%%%%%%%%%%%%%%%%%%%%%%%%%%%%%%%%%%%%%
\subsection{Input u Held Constant when z = 0 and u = -KX when z = 1}
\label{subsection:method1}
%Suppose we have a single-input system:
%\begin{equation}
%    \dot{x}=Ax+Bu
%    \label{eqn:sys}
%\end{equation}
%The new state $z$ indicates the current sensing status. If $z=1$, the feedback control is performed with a predetermined feedback gain $K$ satisfies that both $(A,0)$ and $(A,B)$ are stabilizable. In the reverse case, the sensors are disabled and the control signal is maintained at the last value. 
In this scenario, we formulate the $u$ as follows:
\begin{equation}
    u=-KX \text{ when } z=1 \; \text{and} \;
    \dot{u}=0 \text{ when } z=0 
    \label{eq_z0udot0}
\end{equation}
where $K$ is a predefined gain matrix. From Eq.(\ref{eq_z0udot0}), the derivative of $u$ can be written as,
\begin{equation}
\dot{u} = -zK(AX+Bu)
\end{equation}
We have the augmented state vector $\overline{X}=\left[X\;u\right]^T$ and the following augmented state equation,
\begin{equation}
\dot{\overline{X}} = f(\overline{X}, z) = \left[ 
\begin{array}{rr}
A & B \\
-zKA & \,\, -zKB
\end{array}
\right] \overline{X}
\label{eq_augxbar}
\end{equation}
The objective is to minimize the cost functional,
\begin{equation}
J = \int_{t_0}^{t_f} \!\!\! \left( X^T Q X + sz \right) dt
\label{eq_costf2}
\end{equation} 
where $Q$ is symmetric positive semi-definite, and $s > 0$. We clarify that since a predefined state feedback control is used, Case 1 in Eq.(\ref{eq_z1cases}), hence $J$ does not include a $u^TRu$, $R > 0$ term. Equation (\ref{eq_costf2}) can be expressed as,
\begin{equation}
J = \int_{t_0}^{t_f} \!\!\! \left( \overline{X}^T \overline{Q} \overline{X} + sz \right) dt, \quad
\overline{Q}=\left[
\begin{array}{cc}
	Q & 0 \\
	0 & 0
\end{array}\right]
\label{eq_costf2_a}
\end{equation}
We note in Eq.(\ref{eq_costf2}) that the term $\int_{t_0}^{t_f} sz dt$ imposes a sensing-cost. It is distinct from control cost $\int_{t_0}^{t_f} u^T R u \, dt$ in that $s$ adds to the cost whenever $z = 1$, irrespective of the control input $u$. The Hamiltonian is:
\begin{equation}
    \label{eqn:Hamiltonian}
    H = L + \lambda^T f 
    = \overline{X}^T \overline{Q} \overline{X} + sz + \lambda^T
    \left[ 
        \begin{array}{rr}
        A & B \\
        -zKA & \,\,-zKB
        \end{array}
    \right] \overline{X}
\end{equation}
In Eq.(\ref{eqn:Hamiltonian}), $\lambda=\left[ \lambda_X  \;\lambda_u \right]^T$ is an ($n+m$)-element vector of Lagrange multipliers. 
%The optimal solution can be found by solving the following equation:
%\begin{equation}
%    \label{eqn:optimal_z}
%    \begin{aligned}
%    \frac{\partial H}{\partial z}=0  &\Leftrightarrow 2Sz-\lambda_uK(Ax+Bu)=0 \\
%    &\Leftrightarrow z = \frac{\lambda_uK(Ax+Bu)}{2S}    
%    \end{aligned}
%\end{equation}
%The solution in Eq. \ref{eqn:optimal_z} is infeasible because z can only be either 0 or 1. The Pontryagin's maximum (or minimum) principal (\cite{kirk1970optimal}) can be used to overcome this problem. The solution $z=0$ is optimal if and only if:
To determine a condition for switching $z$, we apply the Pontryagin's minimum principle (\cite{kirk1970optimal, Bryson1975, Stengel94, Lewis2012}), which is a necessary condition for optimality. In the context of this problem, it states that the Hamiltonian must be minimized over the admissible region of $z$ for optimal values of the state and co-state. Since $z$ is constrained to assume discrete integer values, the Pontryagin's minimum principle cannot be directly applied. To address this issue, we allow $z$ to vary continuously in the admissible range $\left[ 0, 1 \right]$. Thereby, we embed our switched system within a larger family of systems, as done in \cite{das2008switched}. For the admissible range of $z$, i.e., $z \in \left[0, 1\right]$, a necessary condition for optimality is found by minimizing the Hamiltonian $H$, Eq.(\ref{eqn:Hamiltonian}). Thus, 
\begin{equation}
    H(\overline{X}^*, \lambda^*, z^*) \le H(\overline{X}^*, \lambda^*, z) \quad \forall \, t \in \left[ t_0, t_f \right]
\end{equation}
for all admissible $z$. We note from Eq.(\ref{eqn:Hamiltonian}) that,
\begin{equation}
    H = \overline{X}^T\overline{Q}\overline{X} + sz + \lambda_X^T \left[ A \quad B\right]\overline{X} - z \lambda_u^T K \left[ A \quad B\right]\overline{X}
\end{equation}
is linear in $z$ for any $\overline{X}^*$ and $\lambda^*$. Therefore, $z^*$ lies at the discrete points $\left\{0, 1\right\}$ for all $t \in \left[ t_0, t_f \right]$, and the formulation of $u$ in Eq.(\ref{eq_z0udot0}) is valid at all instants. The solution $z^*=0$ minimizes $H$ if,
\begin{equation}
    \begin{aligned}
    &\quad\,\left.H(\overline{X}^*, \lambda^*, z)\right|_{z=0} < \left.H(\overline{X}^*, \lambda^*, z)\right|_{z=1} \\
    &\Leftrightarrow  \overline{X}^{*T}\overline{Q}\overline{X}^*+\lambda_X^{*T}\left[ A \quad B\right]\overline{X}^* < \cdots\\
    &\quad\;\; \overline{X}^{*T}\overline{Q}\overline{X}^*+s+\lambda_X^{*T}\left[ A \quad B\right]\overline{X}^*-\lambda_u^{*T}K\left[ A \quad B\right]\overline{X}^*\\
    &\Leftrightarrow \lambda_u^{*T}K\left[ A \quad B\right]\overline{X}^* < s \; \Leftrightarrow \; \lambda_u^{*T}K(AX^*+Bu^*) < s
    \end{aligned}
    \label{eq_sw1a}
\end{equation}
Similarly, $z^*=1$ minimizes $H$ if $\lambda_u^{*T}K(AX^*+Bu^*) \ge s$. Thus, in the sense of minimizing $H$ for all admissible $z$, the condition for optimal switching between sensing ($z = 1$) and no-sensing ($z = 0$) is,
\begin{equation}
    z^* = \left\{ 
    \begin{array}{lll}
    0 & \quad \text{if} \quad & \lambda_u^TK(AX+Bu) < s \\
    1 & \quad \text{if} \quad & \lambda_u^TK(AX+Bu) \ge s
    \end{array}
    \right.
\label{eq_sw1}
\end{equation}
We note that Eq.(\ref{eq_sw1}) is a necessary condition for optimality. The co-states equation and boundary conditions are:
\begin{equation}
    \dot{\lambda}=-\frac{\partial H}{\partial \overline{X}}=-2\overline{Q}\overline{X}-\left[ 
\begin{array}{rr}
A^T & -zA^TK^T \\
B^T & \,\,-zB^TK^T
\end{array}
\right] \lambda
\label{eq_cos1}
\end{equation}
\begin{equation}
    \overline{X}(t_0) = \left[ X_0^T \,\,\; -X_0^TK^T \right]^T, \quad \lambda(t_f)=\frac{\partial \phi}{\partial X}(t_f)=[0]
\label{eq_bc1}
\end{equation}
%Because the initial condition is unknown and the differential equations are nonlinear, there is no general solution for this problem. However, this two-point boundary value problem (TPBVP) can be solved by numerical methods such as "shooting" method or relaxation method (\cite{das2003optimal}).
where $\phi(X(t_f), t_f)$ is the final weighting function, which is set to zero in this study.  Equations (\ref{eq_augxbar}), (\ref{eq_sw1}), (\ref{eq_cos1}) and (\ref{eq_bc1}) form a Two-Point Boundary Value Problem (TPBVP) that can be solved by numerical methods such as the ``Shooting Method" or the ``Relaxation Method" (\cite{Press1992, das2003optimal}). 

%%%%%%%%%%%%%%%%%%%%%%%%%%%%%%%%%%%%%%%%%%%%%%%%%%%%%%%%%%%%%%%%%
\subsubsection*{Observations}
We note the possibility of {\it singular intervals}, \cite{kirk1970optimal, Stengel94}, if the criterion,
\begin{equation}
    \lambda_u^TK(AX+Bu) = s
\label{eq_obs1}
\end{equation}
is satisfied over a finite interval $t\in\left[t_1, t_2\right]$. Another necessary condition for optimality is that if the Hamiltonian is not an explicit function of time, then it is constant along the optimal trajectory for a fixed final-time problem, \cite{kirk1970optimal, Stengel94}. This necessary condition implies,
\begin{equation}
\begin{aligned}
    \dot{H} & = \left( \frac{\partial H}{\partial\overline{X}} \right)^T\dot{\overline{X}} + \left(\frac{\partial H}{\partial \lambda}\right)^T\dot{\lambda} + \frac{\partial H}{\partial z}\dot{z} \\
    &= -\dot{\lambda}^T\dot{\overline{X}} + \dot{\overline{X}}^T\dot{\lambda} + \frac{\partial H}{\partial z}\dot{z} \\
    & = \left[ s - \lambda_u^TK(AX+Bu)\right]\dot{z} = 0
\end{aligned}
\label{eq_obs2}
\end{equation}
From Eqs.(\ref{eq_obs1}) and (\ref{eq_obs2}), we conclude that outside of singular intervals, $\dot{z} = 0$ along optimal trajectories. This confirms that for the proposed problem, where $z \in \left[0, 1\right]$, optimality occurs necessarily at the discrete points, $z \in \left\{ 0, 1 \right\}$ outside of singular intervals. From Eq.(\ref{eq_obs1}), a necessary condition for the existence of a singular interval is,
\begin{equation}
\begin{aligned}
    &\qquad \frac{d}{dt}\left( s - \lambda_u^TK(AX+Bu) \right) = 0 \Rightarrow \left[ \lambda_x^T BK - \lambda_u^T KA \right]\dot{X} = 0
\end{aligned}
\label{eq_obs3}
\end{equation}
over a finite interval $t\in\left[t_1, t_2\right]$. While the above condition does not provide much insight into the singular input $z$, for a scalar system further deductions can be made. For a scalar system, if we represent $A$, $B$, $K$, and $X$ by their scalar equivalents $a$, $b$, $k$, and $x$, then Eq.(\ref{eq_obs3}) reduces to 
\begin{equation*}
    \left( \lambda_x b - \lambda_u a \right) \dot{x} = 0
\end{equation*}
which implies a necessary condition for the existence of singular arcs is either $( \lambda_x b - \lambda_u a ) = 0$, or $\dot{x} = 0$, or both over a finite interval $t\in\left[t_1, t_2\right]$. The former condition can be further differentiated with respect to time to yield the following necessary condition:
\begin{equation}
    \frac{d}{dt}\left( \lambda_x b - \lambda_u a \right) = 0 \; \Rightarrow \; -2bqx = 0
\label{eq_obs4}
\end{equation}
where $q$ is the scalar equivalent of $Q$ defined in Eq.(\ref{eq_costf2}). Assuming the parameters $a$, $b$, $q$, $k$ to be non-zero, Eq.(\ref{eq_obs4}) implies $x = 0$ in $t\in\left[t_1, t_2\right]$, which implies $\dot{x} = 0$, $\ddot{x} = 0$, and so on for higher derivatives in $t\in\left[t_1, t_2\right]$. Equation (\ref{eq_obs4}) still does not produce any condition on $z$ for a singular solution. However, upon differentiation of $\dot{x}$, we have,
\begin{equation}
    \ddot{x} = (a - bkz) \dot{x}
\label{eq_obs5}
\end{equation}
which implies $z^* = a/(bk)$ is a necessary condition for optimality in a singular interval $t\in\left[t_1, t_2\right]$. Thus, for the scalar version, limiting $z \in \left\{ 0, 1 \right\}$ (i.e., sensing vs. no-sensing conditions) in a singular interval may yield sub-optimal solutions. 

A second observation is regarding the inclusion of input/control cost $u^TRu$, $R > 0$, in Eq.(\ref{eq_costf2}). If done, then in Eq.(\ref{eq_costf2_a}) $\bar{Q}~=~\left[ Q \,\, 0 ;\, 0 \,\, R\right]$, and the derivation from Eq.(\ref{eq_sw1a}) to Eq.(\ref{eq_bc1}) would remain unchanged. However, as mentioned after Eq.(\ref{eq_costf2}), since a predefined state feedback control is used, we do not include this term in $J$. Simulations reveal that if incorporated, it complicates the convergence of TPBVP solutions, and hence requires further investigations.

%%%%%%%%%%%%%%%%%%%%%%%%%%%%%%%%%%%%%%%%%%%%%%%%%%%%%%%%%%%%%%%%%%%%
\subsection{Input u = [0] when z = 0 and u = -KX when z~=~1}
\label{subsection:method2}
The method presented in section \ref{subsection:method1} has a disadvantage in that it increases the order of the state equation by $m$. Therefore, the complexity of calculations also increases. This is avoided by imposing the control input to be zero when the states are not sensed, i.e., $u = [0]$ when $z = 0$, as in Case 2 of Eq.(\ref{eq_z0cases}). The objective is to minimize the cost functional of Eq.(\ref{eq_costf2}), since $u$ is still a predetermined input. The state equation takes the form,
\begin{equation}
    \dot{X} = f(X, z) = (A -zBK)X
    \label{eq_st2}
\end{equation}
and the Hamiltonian is,
\begin{equation}
    \label{eqn:Hamiltonian2}
    H=L+\lambda^Tf =X^TQX+sz+\lambda^T\left( A - zBK \right) X
\end{equation}
As in Section \ref{subsection:method1}, for admissible values of $z$, i.e. $z \in \{0, 1\}$, the optimal $z$ is found by first embedding our switched system within a larger family of systems where $z \in [0, 1]$. Thereafter, we minimize the Hamiltonian $H$, Eq.(\ref{eqn:Hamiltonian2}), according to the Pontryagin's minimum principle. Proceeding as in Section \ref{subsection:method1}, the condition for optimal switching between sensing ($z = 1$) and no-sensing ($z = 0$) is,
\begin{equation}
    z^* = \left\{ 
    \begin{array}{lll}
    0 & \quad \text{if} \quad & \lambda^TBKX < s \\
    1 & \quad \text{if} \quad & \lambda^TBKX \ge s
    \end{array}
    \right.
\label{eq_sw2}
\end{equation}
The co-states equation and boundary conditions are:
\begin{equation}
    \dot{\lambda}=-\frac{\partial H}{\partial X}=-2QX - \left( A - zBK\right)^T \lambda
\label{eq_cos2}
\end{equation}
\begin{equation}
    X(t_0) = X_0, \quad \lambda(t_f)=\frac{\partial \phi}{\partial X}(t_f)=[0]
\label{eq_bc2}
\end{equation}
Equations (\ref{eq_st2}), (\ref{eq_sw2}), (\ref{eq_cos2}) and (\ref{eq_bc2}) form a TPBVP that can be solved numerically. 

As in Section \ref{subsection:method1}, a singular interval occurs if $\lambda^TBKX = s$ over a finite interval $t\in\left[t_1, t_2\right]$. A necessary condition for the existence of a singular interval is therefore,
\begin{equation}
    \frac{d}{dt}\left( \lambda^TBKX\right) = 2X^TQBKX + \lambda^T\left( ABK - BKA \right)X = 0
\label{eq_obs6}
\end{equation}
Again, Eq.(\ref{eq_obs6}) does not provide insight into the singular input $z$. Similar to Section \ref{subsection:method1}, further deductions can however be made for the scalar case. Representing $A$, $B$, $K$, $Q$ and $X$ by their scalar equivalents, $a$, $b$, $k$, $q$, and $x$, Eq.(\ref{eq_obs6}) reduces to $2qbkx^2 = 0$. Therefore, the necessary conditions for the existence of a singular interval are, $x = 0$, $\dot{x} = 0$, $\ddot{x} = 0$, and so on, in that interval. We note from Eq.(\ref{eq_st2}) that $z^* = a/(bk)$ ensures $\dot{x} = 0$, and is thus a necessary condition for optimality in singular intervals.  Thus, for the scalar version, limiting $z \in \left\{ 0, 1 \right\}$, as proposed in Eq.(\ref{eq_sw2}), albeit practical, may yield sub-optimal solutions in a singular interval. Finally, an input/control cost $u^TRu$, $R > 0$, can be incorporated in Eq.(\ref{eq_costf2}) by modifying $J$ as follows:
\begin{equation}
    J = \int_{t_0}^{t_f} \!\!\! \left( X^T Q X + u^TRu + sz \right) dt
    = \int_{t_0}^{t_f} \!\!\! \left[ X^T \left( Q + z Q_2 \right) X + sz \right] dt
\label{eq_obs7}
\end{equation}
where, $Q_2 = K^TRK$. However, as mentioned just before Eq.(\ref{eq_st2}), we have not considered the control cost since $u$ is a predefined input. As in Section \ref{subsection:method1}, here also simulations reveal that if included, this term complicates the convergence of TPBVP solutions. This will be an area of future investigations.

%%%%%%%%%%%%%%%%%%%%%%%%%%%%%%%%%%%%%%%%%%%%%%%%%%%%%%%%%%%%%%%%%%%%
\subsection{Input u = [0] when z = 0 and Optimal Input u~=~u* when z = 1}
\label{subsection:method3}
Here, we continue to impose $u$ in Eq.(\ref{eqn:sys}) to be set at $u = [0]$ when the sensors are disabled (i.e., $z=0$). This corresponds to Case 2 of Eq.(\ref{eq_z0cases}), as in Section \ref{subsection:method2}. However, unlike Section \ref{subsection:method2}, the objective here is to use an optimal control $u^*$ when $z=1$. This refers to Case 2 of Eq.(\ref{eq_z1cases}). The state equation takes the form,
\begin{equation}
    \dot{X} = AX + zBu
    \label{eq_st3}
\end{equation}
The goal is to minimize the cost functional of Eq.(\ref{eq_costf1}), with the specific form,
\begin{equation}
J = \int_{t_0}^{t_f} \!\!\! \left( X^T Q X + u^TRu + sz \right) dt
\label{eq_costf4}
\end{equation} 
%Because $u$ and $z$ are both scalars, Eq.(\ref{eqn:cost_func_m2}) can be rewritten as:
%\begin{equation}
%    J = \int_{0}^{t_f}(X^TQX+ru^2+sz^2)dt    
%\end{equation}
From Eqs.(\ref{eq_st3}) and (\ref{eq_costf4}), the Hamiltonian is,
\begin{equation}
    H = L + \lambda^T f = X^T Q X + u^T R u + sz + \lambda^T(AX+zBu)
\label{eq_ham3}
\end{equation}
%We first consider $R$ in Eq.(\ref{eq_costf1}) to be diagonal, i.e. $R = diag(r_i), \; i = 1,2, \cdots m$, in addition to being positive definite. 
As in Sections \ref{subsection:method1} and \ref{subsection:method2}, for admissible values of $z$, i.e. $z \in \{0, 1\}$, the optimal $z$ is found by first embedding our switched system within a larger family of systems where $z \in [0, 1]$. Thereafter, we minimize $H$, Eq.(\ref{eq_ham3}), according to the Pontryagin's minimum principle. The Hamiltonian is quadratic in the unconstrained input $u$. 
%Since $R \in \mathbb{R}^{m \times m}$ is symmetric positive definite, therefore from the \textit{Rayleigh-Ritz inequality}, \cite{Rugh1996}, the term $u^TRu$ in Eq.(\ref{eq_ham3}) satisfies,
%\begin{equation*}
%    0 < \lambda_{min}(R)u^Tu \le u^TRu \le \lambda_{max}(R)u^Tu
%\end{equation*}
%This implies that for a specific value of $z$, $H$ has a minimum for all $u \in \mathbb{R}^m$ when,
Since $R \in \mathbb{R}^{m \times m}$ is symmetric positive definite, therefore for a specific value of $z$, $H$ has a minimum for all $u \in \mathbb{R}^m$ when,
\begin{equation}
    \frac{\partial H}{\partial u} = 2Ru + zB^T\lambda = [0] \; \Rightarrow \; u_{min} = -\frac{1}{2}zR^{-1}B^T\lambda
    \label{eq_umin}
\end{equation}
Now, substituting the above expression of $u = u_{min}$ in Eq.(\ref{eq_ham3}), we have,
\begin{equation}
    H \vert_{u_{min}} = X^T Q X + \lambda^TAX + sz - \frac{1}{4}z^2 \lambda^T B R^{-1} B^T \lambda
\end{equation}
Note that $H \vert_{u_{min}}$ is a quadratic expression 
in $z$. Furthermore, $\frac{1}{4}\lambda^T B R^{-1} B^T \lambda > 0$ for all $\lambda \ne [0]$, since it has a quadratic form and $R^{-1} > 0$ as $R$ is positive definite and symmetric. Hence, $H \vert_{u_{min}}$ has a maximum at,
\begin{equation*}
    \frac{\partial H \vert_{u_{min}}}{\partial z} = 0 \; \Rightarrow \; z_{max} = \frac{2s}{\lambda^T B R^{-1} B^T \lambda} > 0
\end{equation*}
and is monotonically decreasing on either side of $z_{max}$. Therefore, within the admissible range $z \in [0, 1]$, $H$ is minimized at one of the discrete values $z \in \{0, 1\}$. Thus, the optimal input, from Eq.(\ref{eq_umin}), is,
\begin{equation}
    u^* = u_{min} = -\frac{1}{2}zR^{-1}B^T\lambda
    \label{eq_uopt}
\end{equation}
For choosing $z = z^*$, note that the necessary condition for optimality from the Pontryagin's minimum principle is,
\begin{equation*}
    H(X^*, \lambda^*, u^*, z^*) \le H(X^*, \lambda^*, u^*, z) \quad \forall \, t \in \left[ t_0, t_f \right]
\end{equation*}
Hence, we choose $z^* = 0$ when,
\begin{equation}
    \label{eq_sw3}
    \begin{aligned}
        &\quad\left.H\right|_{u_{min}, z=0}<\left.H\right|_{u_{min}, z=1}\\
        &\Leftrightarrow 0 <  s -  \frac{1}{4} \lambda^{*T} B R^{-1} B^T \lambda^* \Leftrightarrow s > \frac{1}{4}\lambda^{*T} B R^{-1} B^T \lambda^*
    \end{aligned}
\end{equation}
Thus, the condition for optimal switching between sensing ($z~=~1$) and no-sensing ($z = 0$) is,
\begin{equation}
    z^* = \left\{ 
    \begin{array}{lll}
    0 & \quad \text{if} \quad & 0.25\lambda^TBR^{-1} B^T \lambda < s \\
    1 & \quad \text{if} \quad & 0.25\lambda^TBR^{-1} B^T \lambda \ge s
    \end{array}
    \right.
\label{eq_sw3a}
\end{equation}
Note that unlike in Sections \ref{subsection:method1} and \ref{subsection:method2}, this scenario does not permit singular intervals. The co-states equation and boundary condition are:
\begin{equation}
    \dot{\lambda}=-\frac{\partial H}{\partial X}=-2QX-A^T\lambda
    \label{eq_cos3}
\end{equation}
\begin{equation}
    X(t_0) = X_0, \quad \lambda(t_f)=\frac{\partial \phi}{\partial X}(t_f)= [0]
    \label{eq_bc3}
\end{equation}
Equations (\ref{eq_st3}), (\ref{eq_uopt}), (\ref{eq_sw3a}), (\ref{eq_cos3}) and (\ref{eq_bc3}) form a TPBVP that can be solved numerically. In Section \ref{section:Examples}, we provide examples of the problems discussed in this section. We present numerical solutions as well as a practical ``shrinking horizon implementation." In the next section, we analyze the sensing-cost problem for a first-order system.

%%%%%%%%%%%%%%%%%%%%%%%%%%%%%%%%%%%%%%%%%%%%%%%%%%%%%%%%%%%%
\section{Infinite Time Case with Single Switching Instant}
\label{sec_1sw}
Scenarios where the sensing-cost problem admits only one switching instant under an infinite time horizon are considered in this section. Specifically, these scenarios are considered within the context of the problem formulation of Section \ref{subsection:method2}, and analytical results are derived. Here, the optimal switching condition in given in Eq.(\ref{eq_sw2}). Define the switching function $\alpha(t)~=~\lambda^T(t)BKX(t)$, noting that $\alpha(t)$ is scalar. The time derivative of $\alpha(t)$ is:
\begin{equation}
\label{eq_Itss1}
    \alpha = \lambda^TBKX \Rightarrow \dot{\alpha} = \left[\frac{\partial \alpha}{\partial \lambda} \right]^T\!\!\!\!\dot{\lambda} + \left[\frac{\partial \alpha}{\partial X} \right]^T\!\!\!\!\dot{X} = (BKX)^T\dot{\lambda} + (\lambda^T BK)\dot{X}
\end{equation}
Substitute equations (\ref{eq_st2}) and (\ref{eq_cos2}) into (\ref{eq_Itss1}), we have:
\begin{equation}
\label{eq_Itss2}
    \begin{aligned}
        \dot{\alpha} &= X^TK^TB^T \left[- 2QX - \left(A - zBK\right)^T\lambda\right] + \lambda^T BK (A-zBK) X \\
        &= -2X^TK^TB^TQX - X^TK^TB^TA^T\lambda + zX^T\left(K^TB^T\right)^2\lambda \\
        &\quad + \lambda^T BKAX - z\lambda^T (BK)^2 X \\
        &= -2X^T(QBK)X + \lambda^T (BKA - ABK)X
    \end{aligned}
\end{equation}
From the equation above, we observe that if $QBK$ is positive (or negative) definite (or semidefinite) and if $A$ and $BK$ commute, the optimal solution will have at most one switching point. In this regard, we state and prove the following Lemma,
\begin{lemma}
    For the LTI system of Eq.(\ref{eqn:sys}) with a single input, i.e., $u \in \mathbb{R}^{1}$, for any $K \in \mathbb{R}^{1\times n}$, a necessary condition for the matrices $A$ and $BK$ to commute is that the controllability matrix $C~=~\left[ B \,\, AB \,\, \cdots \,\, A^{n-1}B\right]$ has a rank of 1.
\end{lemma}
\textbf{Proof:} Note that,
\begin{equation}
    BKB = B(KB) = \rho B, \qquad \rho = KB
\end{equation}
Thus, $\rho = KB$ is an eigenvalue of $BK$ and $B$ is the corresponding eigenvector. Since $C$ is rank 1, therefore,
\begin{equation}
    AB = \bar{\rho}B
\end{equation}
Thus, $\bar{\rho}$ is an eigenvalue of $A$ and $B$ is the corresponding eigenvector. Hence, the matrices $A$ and $BK$ have corresponding eigenvalues of $\bar{\rho}$ and $\rho$ with the common eigenvector $B$. Since two commuting matrices must have the same set of eigenvectors, \cite{Horn13}, the above mentioned necessary condition is satisfied by the eigenvector $B$. \hfill $\square$

Let us assume that the following three conditions are valid,
\begin{enumerate}
    \item $A$ and $(A-BK)$ are Hurwitz
    \item $A$ and $BK$ commute
    \item $QBK$ is positive semidefinite
\end{enumerate}
Then, from Eq.(\ref{eq_Itss2}), it is guaranteed that the optimal solution has at most one switching point. We denote the switching instant as $t_{sw}$. From Eq.(\ref{eq_st2}), for $t \leq t_{sw}$, the system equation is $\dot{X} = (A-BK)X$ and the solution is, $X(t)=e^{(A-BK)t}X_0 \ \forall t \in \left[0, t_{sw} \right]$. The system switches to $\dot X = AX$ at $t = t_{sw}$, the solution in this interval is:
\begin{equation}
\label{eqn:multidim_Xt}
    X(t) = e^{A(t-t_{sw})}X(t_{sw}) = e^{A(t-t_{sw})}e^{(A-BK)t_{sw}}X_0
\end{equation}
We have $e^{A(t-t_{sw})}=e^{At}e^{-At_{sw}}$, and because $A(A-BK) = A^2 - ABK = A^2 - BKA = (A-BK)A$ due to condition (2), we have from Eq.(\ref{eqn:multidim_Xt}),
\begin{equation}
\label{eqn:multidim_Xt1}
    X(t) = e^{At}e^{-At_{sw}}e^{(A-BK)t_{sw}}X_0 = e^{At}e^{-BKt_{sw}}X_0
\end{equation}
Equation (\ref{eq_Itss2}) now becomes,
\begin{equation}
    \dot \alpha = -2X^T(t)(QBK)X(t)
\label{eqn:multidim_Xt1a}
\end{equation}
Integrating $\dot{\alpha}$ from $t_{sw}$ to $\infty$, we have,
\begin{equation}
\label{eqn:integral_of_alphadot}
    \int_{t_{sw}}^\infty{-2X^T(t)(QBK)X(t)}dt=\alpha(\infty)-\alpha(t_{sw})
\end{equation}
From Eqs.(\ref{eq_sw2}) and (\ref{eq_Itss1}), we have $\alpha(t_{sw})=s$. Furthermore, from Eq.(\ref{eq_bc2}) and because $A$ is Hurwitz, we have,
\begin{equation*}
    \left\{ \begin{aligned}
        &\lim_{t\to\infty}\lambda(t)=\boldsymbol{0}_{n\times 1} \\
        &\lim_{t\to\infty}X(t)=\boldsymbol{0}_{n\times 1} 
    \end{aligned} \right. \Rightarrow \lim_{t\to\infty}\alpha(t)=0
\end{equation*}
Thus, substituting (\ref{eqn:multidim_Xt1}) into (\ref{eqn:integral_of_alphadot}), we have,
\begin{equation}
\label{eqn:integral_of_alphadot_2}
    -2X_0^Te^{-K^TB^Tt_{sw}} \left( \int_{t_{sw}}^{\infty}{e^{A^Tt}(QBK)e^{At}}dt \right) 
    e^{-BKt_{sw}}X_0=-s
\end{equation}
From the Lyapunov equation, because $A$ is Hurwitz and $QBK \ge 0$, there exists a matrix $P \ge 0$ such that $A^TP+PA=-QBK$ \cite{Antsakalis97}. Note that,
\begin{equation}
\begin{aligned}
    \frac{d}{dt}\left(e^{A^Tt}Pe^{At}\right) &= \frac{d}{dt}\left( e^{A^Tt} \right)Pe^{At}
    + e^{A^Tt}P\frac{d}{dt}\left( e^{At}\right)\\
    &= e^{A^Tt} A^T P e^{At} + e^{A^Tt} P A e^{At} \\
    &= - e^{A^Tt} (QBK) e^{At}
\end{aligned}
\end{equation}
Therefore, Eq.(\ref{eqn:integral_of_alphadot_2}) reduces to,
\begin{equation}
\label{eqn:integral_of_alphadot_reduced}
\begin{aligned}
    &\quad-2X_0^Te^{-K^TB^Tt_{sw}}\int_{t_{sw}}^{\infty}d\left(-e^{A^Tt}Pe^{At} \right)
    e^{-BKt_{sw}}X_0=-s\\
    &\Rightarrow -2X_0^Te^{-K^TB^Tt_{sw}} \left[ \boldsymbol{0}_{n\times n} + e^{A^Tt_{sw}}Pe^{At_{sw}} \right]e^{-BKt_{sw}}X_0=-s\\
    &\Rightarrow -2X_0^Te^{-K^TB^Tt_{sw}} e^{A^Tt_{sw}}Pe^{At_{sw}} e^{-BKt_{sw}}X_0=-s\\
    &\Rightarrow 2X_0^Te^{(A-BK)^Tt_{sw}}Pe^{(A-BK)t_{sw}}X_0=s
\end{aligned}
\end{equation}
For a multi-dimensional system, Eq.(\ref{eqn:integral_of_alphadot_reduced}) is a transcendental equation which does not have a closed-form solution. However, since this equation has at most one solution, numerical methods can be used to find its root. We make the following remark:
\begin{remark}
Equation (\ref{eqn:integral_of_alphadot_reduced}) is equivalently, $2X^T(t_{sw})PX(t_{sw})=s$.
\label{rem1}
\end{remark}
The above observation can be made from Eq.(\ref{eqn:multidim_Xt1}). Next, we propose a condition on $s$ that guarantees the existence of a switching time $t_{sw}$. However, we first state and prove the following Lemma:
\begin{lemma}
\label{lemma:strictly_monotonic_of_injective_function}
Let $I$ be a real interval and let $f:I \mapsto \mathbb{R}$ be a continuous injective function. Then $f$ is strictly monotonic.
\end{lemma}
\textbf{Proof:} Suppose that $f$ is not strictly monotonic. Then there exist $a, b, c \in I$ with $a < b < c$ such that either:
\begin{equation*}
f(a) \leq f(b) \ \text{and} \ f(b) \geq f(c)
\end{equation*}
or
\begin{equation*}
f(a) \geq f(b) \ \text{and} \ f(b) \leq f(c).
\end{equation*}
Without loss of generality, assume that $f(a) \leq f(b)$ and $f(b) \geq f(c)$. Since $f$ is injective, the case $f(a)=f(b)$ implies $a=b$, which contradicts $a<b<c$. Hence, we must have $f(a) < f(b)$ and $f(b) > f(c)$. Let,
\begin{equation*}
d = \frac{f(b) + \max\{f(a), f(c)\}}{2}    
\end{equation*}
By the Intermediate Value Theorem, there exist $x_1 \in (a,b)$ and $x_2 \in (b,c)$ such that $f(x_1)=d$ and $f(x_2)=d$. Since $x_1 < b < x_2$ and $f(x_1)=f(x_2)$, this contradicts the injectivity of $f$. Therefore, $f$ must be strictly monotonic. \hfill $\square$

Now, we state and prove the following Proposition:
\begin{proposition}
\label{proposition:condition_of_s}
Equation (\ref{eqn:integral_of_alphadot_reduced}) admits a unique positive root if and only if $0 < s < 2 X_0^T P X_0$.
\end{proposition}
\textbf{Proof:} Let $f(t) = 2X_0^T e^{(A-BK)^T t} P e^{(A-BK)t} X_0$. The function $f(t)$ is continuous in $t$. Furthermore, under the assumption that $QBK$ is positive semidefinite, the system admits at most one switching time. Hence, $f(t)$ is injective on $(0,+\infty)$. By Lemma \ref{lemma:strictly_monotonic_of_injective_function}, $f(t)$ is strictly monotonic on this interval. Since $P$ is positive semidefinite and $(A-BK)$ is Hurwitz, we have:
\begin{equation*}
        \lim_{t \to 0^+} f(t) = 2X_0^T P X_0 \geq 0, \quad
        \lim_{t \to \infty} f(t) = 0.
\end{equation*}
Therefore, if $2X_0^T P X_0 > 0$, the function $f(t)$ is {\it bijective} from $(0,+\infty)$ to $\left(0, 2X_0^T P X_0\right)$. Hence, the equation $f(t)=s$ admits a unique solution if and only if $0 < s < 2 X_0^T P X_0$. \hfill $\square$

We observe that, if $A$ and $(A-BK)$ are Hurwitz, $A$ and $BK$ commute and if $QBK$ is negative semidefinite, then the only optimal control is $z^*(t)=0 \, \forall \, t$. This is because, for a negative semidefinite $QBK$, following the derivation of Eq.(\ref{eqn:integral_of_alphadot_reduced}) and from Remark \ref{rem1}, we have, $2X^T(t_{sw})PX(t_{sw})=-s$, which does not have a solution for $s>0$. Thus, $z^*$ does not undergo any switching. In this case, $z^*(t)=0 \, \forall \, t$ is the only possibility for obtaining a finite cost. We end this section by making the following remark,
\begin{remark}
If $X_0^T P X_0 = 0$, then the optimal solution satisfies $z^*(t)=0 \ \forall t$.
\end{remark}

%%%%%%%%%%%%%%%%%%%%%%%%%%%%%%%%%%%%%%%%%%%%%%%%%%%%%%%%%%%%%%%%%%%%%%%%%%%%%%%%%%%%%%%
\subsection{Closed-Form Solution for First-order System}
\label{sec_1st}
Extending the discussion so far in Section \ref{sec_1sw}, we next explore some interesting characteristics of the optimal problem of Section \ref{subsection:method2} when posed for a first-order system. Consider the first order system,
\begin{equation}
\label{eqn:sys1order}
    \dot{x} = ax+bu
\end{equation}
We assume that the parameters $a$ and $b$, and the initial condition $x(0)=x_0>0$ are known. The control signal is, $u=-kzx$, where $k$ is a constant positive feedback gain and $z$ is the sensing state. From the characteristic of first-order systems, we observe that $x(t) \ne 0 \ \forall t$. The dynamical equation becomes,
\begin{equation}
\label{eqn:sys1order_x}
    \dot{x} = ax - bkzx = (a-bkz)x
\end{equation}
The above structure conforms to the system of Section \ref{subsection:method2}. The objective is to minimize the cost functional,
\begin{equation}
\label{eqn:cost_1storder}
    J=\int_0^{t_f}[(q_1+q_2z)x^2+sz]dt
    %J=\int_0^{t_f}[(q_1+q_2z^2)x^2+sz^2]dt
\end{equation}
where $q_1 \geq 0$, $q_2 \geq 0$, and $s > 0$ are constant. In the equation above, $q_1$ is the coefficient associated with the state cost, while $q_2$ becomes active only when the state is sensed, i.e., when $z=1$. The term $q_2zx^2$ represents an input cost of the form given in Eq.(\ref{eq_obs7}). The Hamiltonian $H$ is ($\lambda$ is the Lagrange multiplier),
%and distinguishes this problem from that of Section \ref{subsection:method2}. 
\begin{equation}
    \begin{aligned}
    H &= (q_1 + q_2z)x^2 + sz+\lambda (ax-bkzx) \\
    %H &= (q_1 + q_2z^2)x^2 + sz^2+\lambda (ax-bkzx)
    &= q_1x^2 + a\lambda x + z(q_2x^2 + s - bk\lambda x)
    \end{aligned}
\end{equation}
 The Pontryagin's minimum principle is applied to minimize $H$ with respect to $z$, by expanding the admissible range of $z$ to the interval $[0,1]$, as done in Section \ref{sect:theory}. This yields,
\begin{equation}
\label{eqn:firstorderMinimumPrinciple}
    \begin{aligned}
        %&\quad\left.H\right|_{z=0}<\left.H\right|_{z=1}\\
        %&\Leftrightarrow q_1x^2+\lambda ax < q_1x^2 + q_2x^2 + s + \lambda x(a-bk)\\
        %&\Leftrightarrow -q_2x^2+bk\lambda x < s
        &\quad H(x^*,\lambda^*,z^*) \leq H(x^*,\lambda^*,z) \\
        &\Leftrightarrow z^*(q_2x^2 + s - bk\lambda x) \leq z(q_2x^2 + s - bk\lambda x) \\
        &\Leftrightarrow z^* = \left\{ \begin{array}{lll}
        0 & \text{if} & \; (q_2x^2 + s - bk\lambda x) > 0 \\
        1 & \text{if} & \; (q_2x^2 + s - bk\lambda x) < 0 \\
        \text{indeterminate} & \text{if} & \; (q_2x^2 + s - bk\lambda x) = 0 \\
        \end{array} \right.\\
    \end{aligned}        
\end{equation}
Denoting $\alpha(t) = q_2x^2(t) + s - bk\lambda(t)x(t)$ to be the \textit{switching function}, we observe that $z^*(t) = 0$ when $\alpha(t) > 0$ and $z^*(t) = 1$ when $\alpha(t) < 0$. The time derivative $\dot{\alpha}(t)$ is,
\begin{equation}
\label{eqn:switching_1order}
    \dot{\alpha}(t) = (2q_2x - bk\lambda)\dot{x} - bkx\dot{\lambda}
\end{equation}
The co-state equation is,
\begin{equation}
\label{eqn:costate_1order}
    \begin{aligned}
    \dot\lambda &= -\frac{\partial H}{\partial x}
    = -2(q_1 + q_2z)x - \lambda (a-bkz)\\
    &= -2(q_1 + q_2z)x - \lambda (a-bkz)
    \end{aligned}    
\end{equation}
Substituting Eqs.(\ref{eqn:sys1order_x}) and (\ref{eqn:costate_1order}) into Eq.(\ref{eqn:switching_1order}), we have,
\begin{equation}
    \label{eqn:alpha_dot}
    \begin{aligned}
    \dot\alpha &= (2q_2x-bk\lambda)(a-bkz)x \\
    &\quad-bkx[-2(q_1 + q_2z)x-\lambda(a-bkz)] \\
    &= 2q_2(a-bkz)x^2-bk\lambda x(a-bkz) +2(q_1 + q_2z)bkx^2 \\
    &\quad + bkx\lambda(a-bkz)\\
    &= 2x^2[q_2(a-bkz)+(q_1+q_2z)bk] = 2(q_1bk+q_2a)x^2
    \end{aligned}
\end{equation}
Note that since the system considered here is scalar, the commutativity requirement of Eq.(\ref{eq_Itss2}) is not applicable in this case. In Eq.(\ref{eqn:alpha_dot}), since $a,b,q_1,q_2,$ and $k$ are constants, it follows that the sign of $\dot\alpha(t)$ remains unchanged over time. Therefore, we can conclude that the optimal solution will have at most one switching point. Depending on the value of ($a,b,k,q_1,q_2$), there are three possible cases listed in Table \ref{table:firstorder}. Note that the switching time $t_{sw}$ depends on $a$, $b$, $k$, $q_1$, $q_2$, and $t_f$.

%Figure \ref{fig:firstorder} visualizes the optimal sequences for each case.
\begin{table}[ht]
\begin{center}
\caption{Optimal sensing sequences - first-order system}
\label{table:firstorder}
\begin{tabular}{|c|c|l|}
\hline
Case & $q_1bk+q_2a$ & Possible optimal sequences \\\hline
1 & $+$ & $\begin{array}{ll} z=0 \ \forall t \\
\text{or} \ z=1 \ \forall t\\
\text{or } z(t) = \left\{ \begin{array}{lll}
         1 & \text{if} & 0 \leq t < t_{sw}  \\
         0 & \text{if} & t_{sw} \leq t
    \end{array} \right.
    \end{array}$\\ \hline
2 & $-$ & $\begin{array}{ll} z=0 \ \forall  t\\
\text{or} \ z=1 \ \forall t\\
\text{or } z(t) = \left\{ \begin{array}{lll}
         0 & \text{if} & 0 \leq t < t_{sw}  \\
         1 & \text{if} & t_{sw} \leq t
    \end{array} \right.
    \end{array}$\\ \hline
3 & $0$ & $\begin{array}{ll} z=0 \ \forall t\\
\text{or} \ z=1 \ \forall t\\
\text{or Singular Control} \end{array}$\\ \hline
\end{tabular}
\end{center}
%\vspace{-7pt}
\end{table}
\begin{comment}
    
\begin{figure}[ht]
\begin{center}
\includegraphics[width = 0.65\columnwidth]{Figures/fig_1storder_cases_v4.eps}
%\vspace{-.3in}
\caption{Switching Condition for the First-Order System \textcolor{red}{(It will be corrected later to match Table 1)}} 
\label{fig:firstorder}
\end{center}
\end{figure}

\end{comment}
We consider a special case when $t_f \to \infty$. In this case, we can prove that a closed-form analytical solution exists. We propose the following:

%%%%%%%%%%%%%%%%%%%%%%%%%%%%%%%%%%%%%%%%%%%%%%%%%%%%%%%%%%%%%%%%%%%%%%%%%%%
% NOTE: Place the figure* here to keep the figures close to the related text.
% Figure: example for 5.1
\begin{figure*}[t]
\begin{center}
\includegraphics[width = 6.0in]{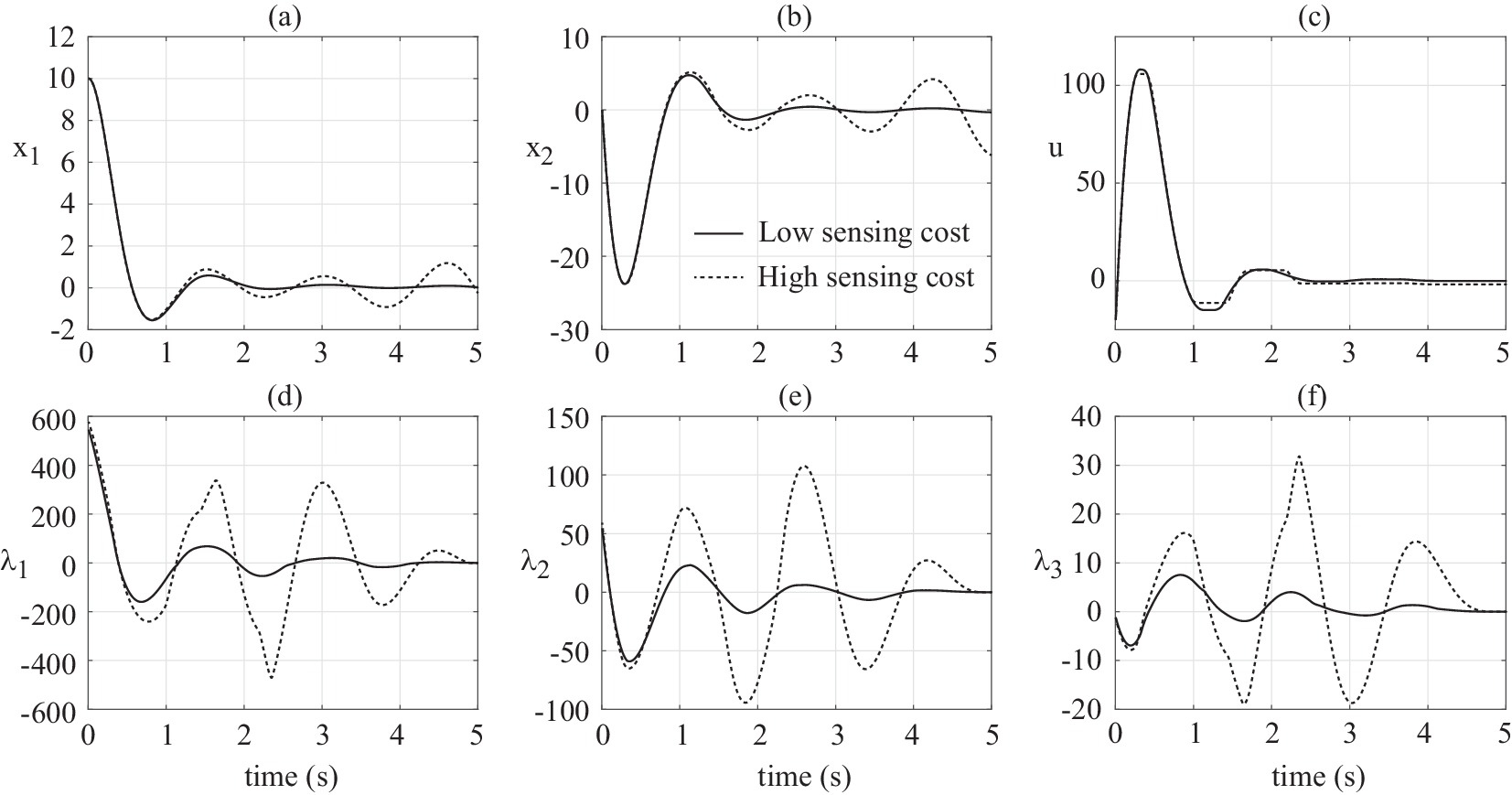}
%\vspace{-.3in}
\caption{States and Co-states for Simulations in Section \ref{section:ex3p1} for Case 1: Low, and Case 2: High Sensing-Cost Weights}
\label{fig:RelaxBvp_3p1_2ndorder}
\end{center}
\end{figure*}
%%%%%%%%%%%%%%%%%%%%%%%%%%%%%%%%%%%%%%%%%%%%%%%%%%%%%%%%%%%%%%%%%%%%%%%%%%%

\begin{proposition}
For the first order system described in Eq.(\ref{eqn:sys1order}), with a closed-loop behavior of Eq.(\ref{eqn:sys1order_x}), a cost functional of Eq.(\ref{eqn:cost_1storder}), and an optimal switching sequence of Eq.(\ref{eqn:firstorderMinimumPrinciple}), the following is true under the assumption that $a < 0$ and $(a - bk) <0$, i.e., the open-loop and closed-loop systems are stable:
\begin{enumerate}
    \item If $(q_1bk+q_2a) > 0$, the optimal control is a bang-bang control that with at most one switching point. Furthermore, if $0 < {-as}/[{(q_1bk+q_2a)x_0^2}] < 1$, the optimal switching time is,
    \begin{equation}
        t_{sw} = \frac{1}{2(a-bk)}ln \left( \frac{-as}{(q_1bk+q_2a)x_0^2} \right)
    \end{equation}
    If the inequality is not satisfied, the optimal control is $z^*(t) = 0 \ \forall t$,
    \item If $q_1bk+q_2a \leq 0$, the optimal control is $z^*(t) = 0 \ \forall t$.
\end{enumerate}
\end{proposition}

\textbf{Proof for (1):} Note from Eqs.(\ref{eqn:sys1order}) and (\ref{eqn:sys1order_x}) that the first order system will not produce $x = 0$ over a finite interval $t \in [t_1, t_2]$. Thus, a necessary condition for the existence of a singular interval, from Eq.(\ref{eqn:alpha_dot}), is $(q_1bk+q_2a) = 0$. Note that a singular interval corresponds to the ``indeterminate" condition of Eq.(\ref{eqn:firstorderMinimumPrinciple}) when valid over an interval $t \in [t_1, t_2]$. Since, in this case, $q_1bk~+~q_2a~>~0$, therefore, a singular interval will not exist. Hence, from Eq.(\ref{eqn:firstorderMinimumPrinciple}), the optimal control in this case is bang-bang control. From Eq.(\ref{eqn:alpha_dot}), $\alpha(t)$ is strictly increasing. Therefore, the optimal control has at most one switching point, as in Case 1 of Table \ref{table:firstorder}. Assume that the switching time exists (denoted by $t_{sw}$), the optimal control can be expressed as,
\begin{equation}
    z^*(t) = \left\{ \begin{array}{lll}
         1 & \text{if} & 0 \leq t < t_{sw}  \\
         0 & \text{if} & t_{sw} \leq t
    \end{array} \right.
\end{equation}
We integrate Eq.(\ref{eqn:sys1order}) to find the expression of $x^*(t)$:
\begin{equation}
    x^*(t) = \left\{
    \begin{array}{ll}
        x_0e^{(a-bk)t} \quad \quad & \text{if} \quad 0\leq t\leq t_{sw}\\
        (x_0e^{-bkt_{sw}})e^{at} \quad & \text{if} \quad t>t_{sw}
    \end{array}
    \right.
\end{equation}
The terminal condition of this problem is $\lambda(\infty)=({\partial \phi}/{\partial x}) = 0$. Since we assume the system to be both open-loop and closed-loop stable, $x(\infty) \to 0$. Combining with the switching condition $\alpha(t_{sw})=0$, we have $\alpha(\infty) = s$. We can solve for the switching time $t_{sw}$:
\begin{equation}
\label{eqn:t_sw_1order}
\begin{aligned}
    &\quad \int_{t_{sw}}^\infty\dot\alpha(t)dt = \alpha(\infty)-\alpha(t_{sw}) = s\\
    &\Rightarrow 2(q_1bk+q_2a)x_0^2e^{-2bkt_{sw}}\int_{t_{sw}}^\infty e^{2at}dt=s\\
    &\Rightarrow 2(q_1bk+q_2a)x_0^2e^{-2bkt_{sw}}\left(-\frac{1}{2a}e^{2at_{sw}}\right)=s\\
    &\Rightarrow e^{2(a-bk)t_{sw}}=\frac{-as}{(q_1bk+q_2a)x_0^2}\\
    &\Rightarrow t_{sw}=\frac{1}{2(a-bk)}ln\left(\frac{-as}{(q_1bk+q_2a)x_0^2}\right)
\end{aligned}
\end{equation}
This equation shows that there is an analytical solution for the first-order infinite-time optimal problem. To guarantee the existence of the one-switching-point solution, the criterion for $t_{sw}$ is $t_{sw}>0$. From Eq.(\ref{eqn:t_sw_1order}), the corresponding condition is $0~<~{-as}/[(q_1bk+q_2a)x_0^2]~<~1$. If the condition does not hold, the optimal control does not contain any switching point. Because the final time $t_f$ tends to infinity, and $\alpha(\infty) = s > 0$, therefore the optimal control in this case is $z^*(t) = 0 \ \forall t$. \hfill $\square$

\begin{comment}
% Previous proof for (2)
\textbf{Proof for (2):} From Table \ref{table:firstorder}, the optimal control in this case can be one of the three solutions: $z(t) = 0 \ \forall t$, $z(t) = 1 \ \forall t$, or $z(t) = 0$ for $0 \leq t < t_{sw}$ and $z(t) = 1$ for $t_{sw} \leq t$. Note that since $t_f \to \infty$, only the solution  $z(t) = 0 \ \forall t$ will yield bounded cost. Therefore, it is the optimal solution. $\square$

% Previous proof for (3)
\textbf{Proof for (3):} If we have singular control on the sub-interval $t \in \left[ t_1,t_2 \right]$, then over that interval,
\begin{equation}
    %\begin{array}
        \frac{\partial H}{\partial z} = 0 \Rightarrow q_2x^2 + s - bk\lambda x = 0
    %\end{array}    
\end{equation}
and the generalized condition for singular control (\cite{Stengel94}) requires:
\begin{equation}
\label{eqn:genConvexCond}
    %\begin{array}
        (-1)^{p/2}\frac{\partial}{\partial z}\left[\frac{d^p}{dt^p}\frac{\partial H}{\partial z}\right] \geq 0
    %\end{array}    
\end{equation}
where $p$ is the order of singularity. From Eq.(\ref{eqn:alpha_dot}), note that $q_1bk+q_2a = 0$, we have:
\begin{equation}
    \begin{array}{l}
        \frac{d}{dt}\frac{\partial H}{\partial z} = 2(q_1bk+q_2a)x^2 = 0 \\
        \Rightarrow \frac{\partial}{\partial z}\left[\frac{d^n}{dt^n}\frac{\partial H}{\partial z}\right] = 0 \quad \forall n > 0
    \end{array}    
\end{equation}
The condition in Eq.(\ref{eqn:genConvexCond}) is satisfied.
\end{comment}

\textbf{Proof for (2):} First, consider the case $q_1bk+q_2a < 0$. From Table \ref{table:firstorder}, the optimal control in this case can be one of the three solutions: $z(t) = 0 \ \forall t$, $z(t) = 1 \ \forall t$, or $z(t) = 0$ for $0 \leq t < t_{sw}$ and $z(t) = 1$ for $t_{sw} \leq t$. Note that since $t_f \to \infty$, only the solution  $z(t) = 0 \ \forall t$ will yield bounded cost. Therefore, it is the optimal solution. 

If $q_1bk+q_2a = 0$, we will prove the singular control does not exist in this case. If we have singular control on the sub-interval $t \in \left[ t_1,t_2 \right]$, then over that interval,
\begin{equation}
\label{eqn:singularcontrol}
    %\begin{array}
        \frac{\partial H}{\partial z} = 0 \Rightarrow q_2x^2 + s - bk\lambda x = 0
    %\end{array}    
\end{equation}
Note that the left-hand side of the equation above is the switching function $\alpha(t)$. From Eq.(\ref{eqn:alpha_dot}), we have $\dot\alpha(t)=2(q_1bk+q_2a)x^2=0 \ \forall t$. Therefore, if the singular control exists, equation (\ref{eqn:singularcontrol}) holds for all $t > 0$. Recall from  Proof of (1) that $\lim_{t\to\infty} x(t) = 0$. Also, $\lim_{t\to\infty} \lambda(t) = 0$ from the final condition on the co-states. Therefore,
\begin{equation}
\lim_{t\to\infty}\alpha(t) = \lim_{t\to\infty}(q_2x^2 + s - bk\lambda x) = s > 0
\end{equation}
This is a contradiction. Therefore, a singular control does not exist, $\alpha(t) = s \ \forall t$ and $z^*(t) = 0 \ \forall t$. \hfill $\square$

%%%%%%%%%%%%%%%%%%%%%%%%%%%%%%%%%%%%%%%%%%%%%%%%%%%%%%%%%%%%%%%%%%%%%%%%%%%
\section{Numerical Solution Examples}
\label{section:Examples}

%%%%%%%%%%%%%%%%%%%%%%%%%%%%%%%%%%%%%%%%%%%%%%%%%%%%%%%%%%%%%%%%%%%%%%%%%%%
\subsection{Section \ref{subsection:method1} Example}
\label{section:ex3p1}

In this example, we choose,
\begin{equation}
    A = \left[ 
    \begin{array}{cc}
    0 & 1 \\
    -16 & 1
    \end{array}
    \right], \; 
    B = \left[ 
    \begin{array}{c}
    0 \\
    1
    \end{array}
    \right], \;
    X(0) = \left[ 
    \begin{array}{c}
    10 \\
    0
    \end{array}
    \right], \;
    K  = \left[ 
    \begin{array}{cc}
    2 & \; 5
    \end{array}
    \right]
    \label{eq_sim1}
\end{equation}

%%%%%%%%%%%%%%%%%%%%%%%%%%%%%%%%%%%%%%%%%%%%%%%%%%%%%%%%%%%%%%%%%%%%%%%%%%%
% Note: place the figure here to keep it close to the related text
% Figure: comparison of control input and sensing state, example 5.1
\begin{figure}[b] % b: place it at the bottom of page
\begin{center}
\includegraphics[width = 0.95\columnwidth]{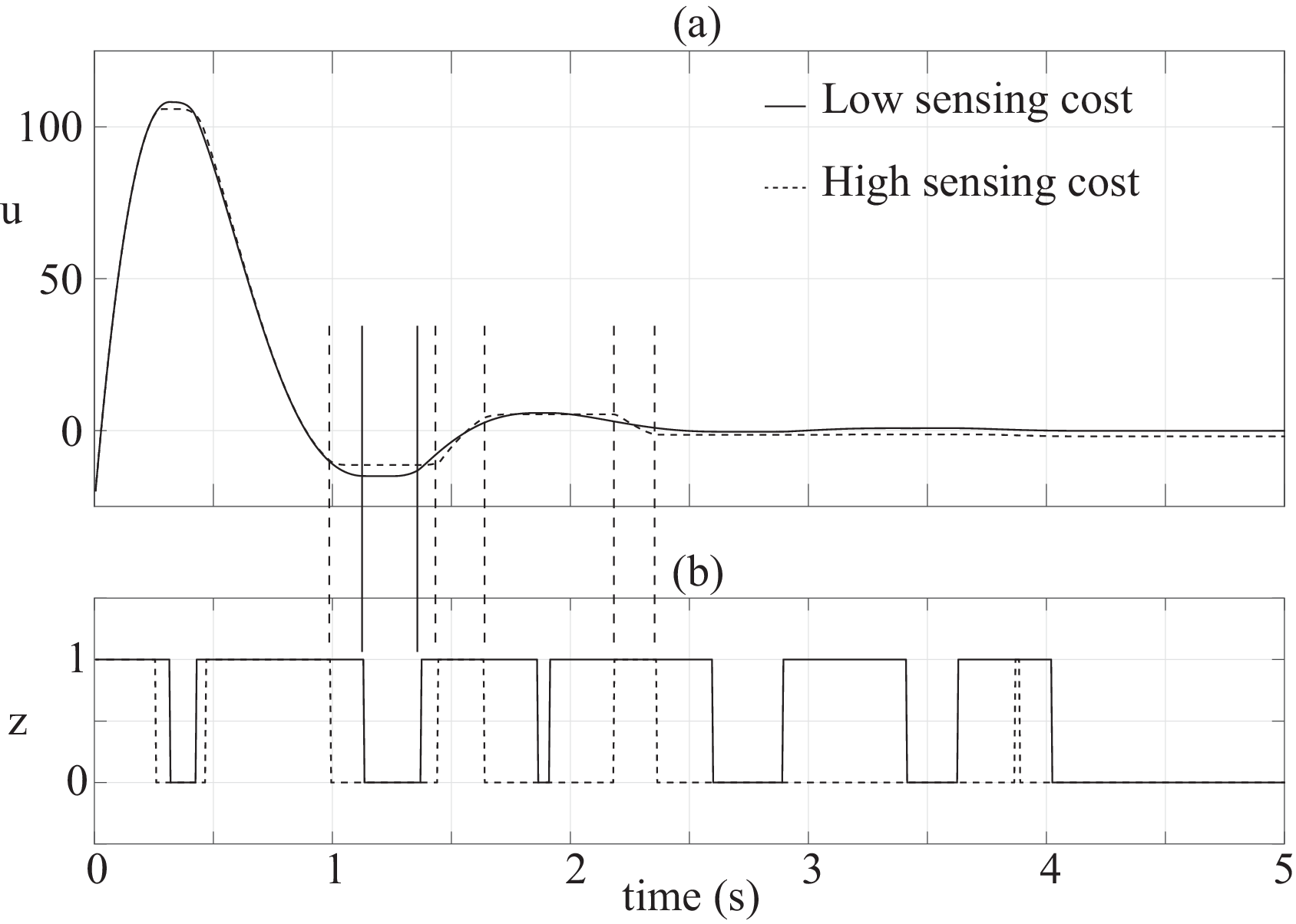}
%\vspace{-.05in}
\caption{Comparison of Control Input and Sensing State} 
\label{fig:RelaxBvp_3p1_2ndorder_sw}
\end{center}
\end{figure}

% Note: place the figure here to keep it close to the related text
%Figure: example 5.2
\begin{figure*}[t]
\begin{center}
\includegraphics[width = 2.0\columnwidth]{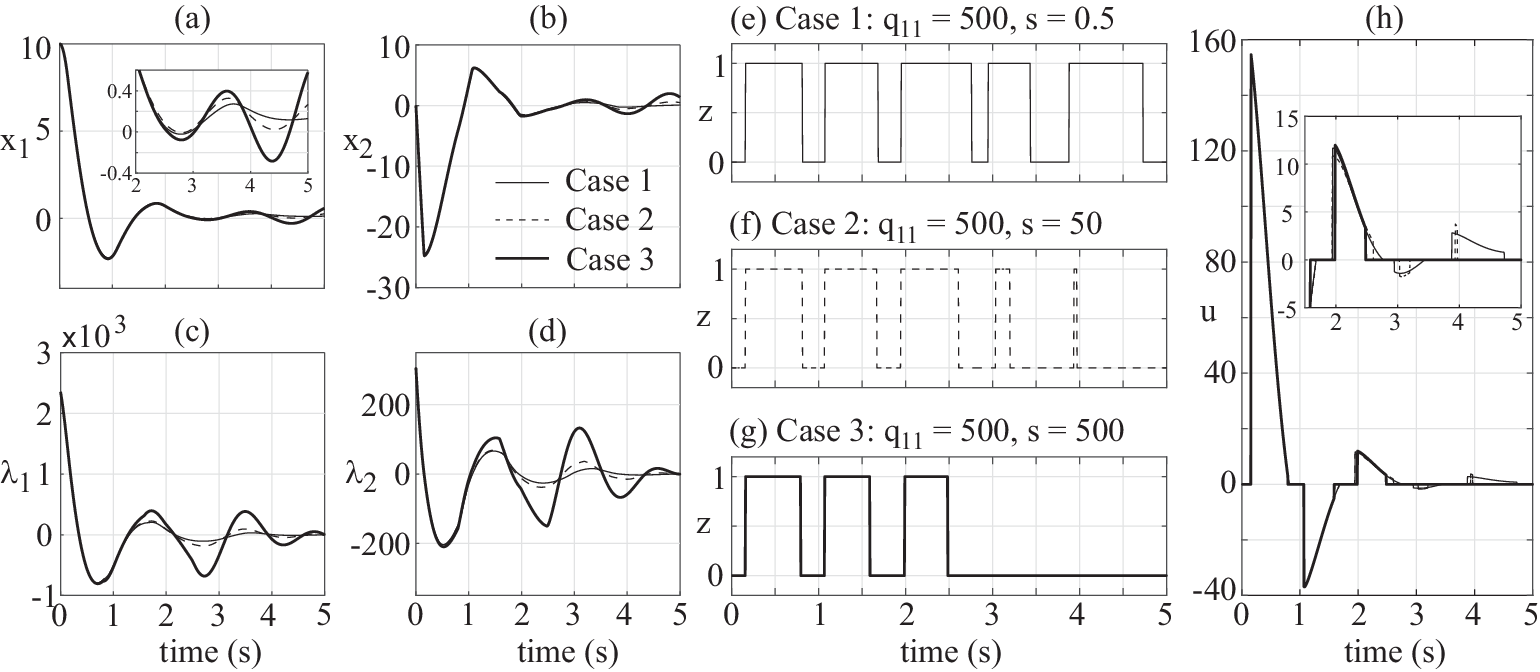}
%\vspace{-.3in}
\caption{Section \ref{section:ex3p2} Simulation: States, Co-states, Switching Variable $z$ and Control Input for Case 1: Low, Case 2: Intermediate, and Case 3: High Sensing-Cost} 
\label{fig:RelaxBvp_3p2_2ndorder}
\end{center}
\end{figure*}
%%%%%%%%%%%%%%%%%%%%%%%%%%%%%%%%%%%%%%%%%%%%%%%%%%%%%%%%%%%%%%%%%%%%%%%%%%%

Definitions of the matrices, and vectors in Eq.(\ref{eq_sim1}) can be found in Section \ref{subsection:method1}. We consider two cases to compare results. In both cases, we use the same $Q$ matrix. However, the sensing-cost weight $s$ is varied. The weights are:
\begin{equation}
    Q = \left[ 
    \begin{array}{cc}
    100 & 0 \\
    0 & 0
    \end{array}
    \right], \; 
    s = \left\{ 
    \begin{array}{rl}
    1 & \; \text{Case 1: Low Sensing-Cost} \\
    1000 & \; \text{Case 2: High Sensing-Cost}
    \end{array}
    \right.
    \label{eq_sim2}
\end{equation}

We recall that in this scenario, $u = -KX$ when $z=1$, and $\dot{u}=0$ when $z=0$. For both cases of Eq.(\ref{eq_sim2}), the TPBVPs were solved using the ``Relaxation Method" \cite{Press1992}. The simulation results are presented in Figs.\ref{fig:RelaxBvp_3p1_2ndorder} and \ref{fig:RelaxBvp_3p1_2ndorder_sw}. 

Figure \ref{fig:RelaxBvp_3p1_2ndorder} gives the state and co-state trajectories for the two cases. Figure \ref{fig:RelaxBvp_3p1_2ndorder_sw} compares the control input $u$ and the sensing/no-sensing switching state $z$ for these cases. From Fig.\ref{fig:RelaxBvp_3p1_2ndorder_sw}(b), it is clear that high sensing-cost leads to a smaller net sensing interval. The net sensing durations are,
\begin{equation}
    \text{Case 1:} \; \int_0^5 z dt = 3.11\text{s}, \; \text{Case 2:} \; \int_0^5 z dt = 1.175\text{s}
\end{equation}

Figure \ref{fig:RelaxBvp_3p1_2ndorder_sw}(a) depicts some intervals when $z = 0$ for both cases, confirming that indeed $\dot{u} = 0$ when $z = 0$. In Eq.(\ref{eq_sim1}), while $A$ is not Hurwitz, $(A-BK)$ is Hurwitz. Hence, we expect that when the sensing-cost is low, i.e. in Case 1, $x_1$ will be relatively more convergent to $0$ than in Case 2, which applies a high sensing-cost. This is confirmed in Fig.\ref{fig:RelaxBvp_3p1_2ndorder}(a), where Case 2 causes $x_1$ to have a divergent trajectory. This divergence is because the stabilizing feedback $u = -KX$ is used for a shorter duration and almost not used after $t = 2.5$s. Figure \ref{fig:RelaxBvp_3p1_2ndorder} also confirms the correct implementation of the boundary conditions for states and co-states, as given in Eq.(\ref{eq_bc1}).

%%%%%%%%%%%%%%%%%%%%%%%%%%%%%%%%%%%%%%%%%%%%%%%%%%%%%%%%%%%%%%%%%%%%%%%%%%%
% Note: place the figure here to keep it close to the related text
%Figure: example 5.3 
\begin{figure*}[ht]
\begin{center}
\includegraphics[width = 2.0\columnwidth]{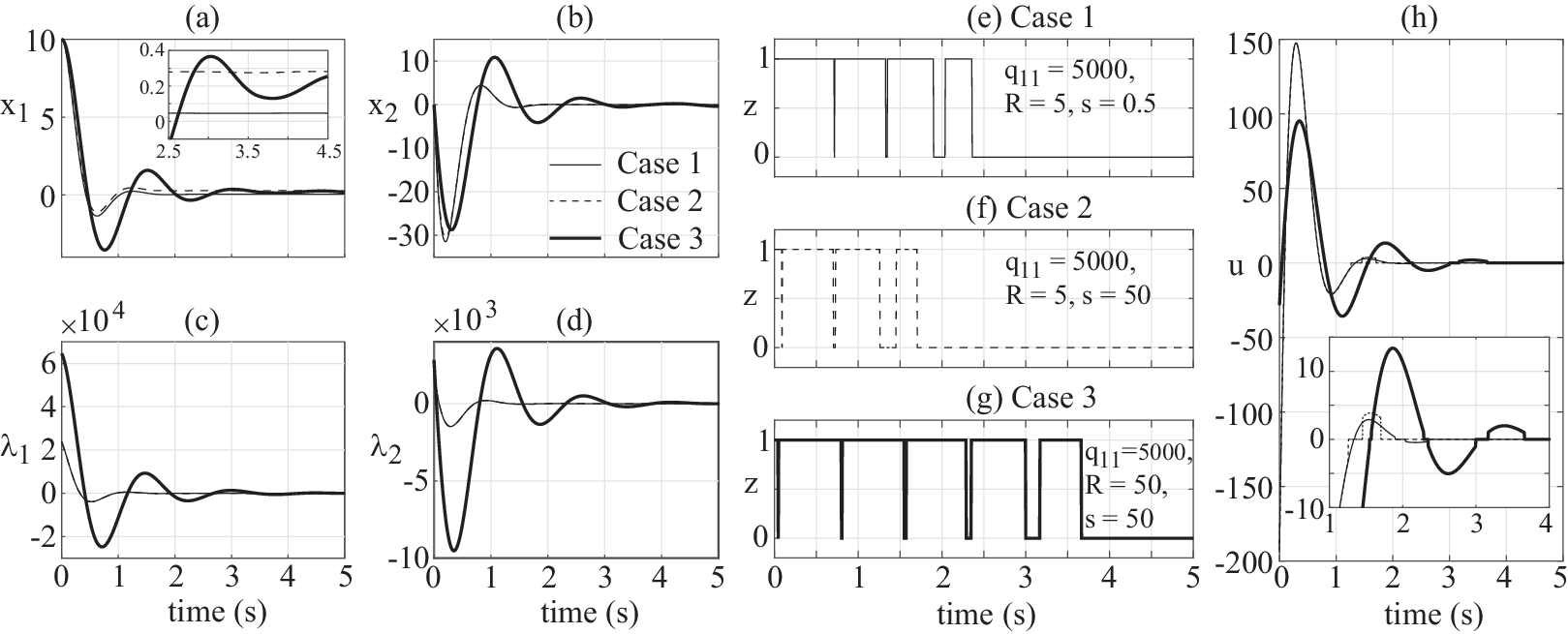}
%\vspace{-.3in}
\caption{Section \ref{section:ex3p3} Simulation: States, Co-states, Switching Variable $z$ and Control Input for Cases: 1, 2 and 3} 
\label{fig:RelaxBvp_3p3_2ndorder}
\end{center}
\end{figure*}

%%%%%%%%%%%%%%%%%%%%%%%%%%%%%%%%%%%%%%%%%%%%%%%%%%%%%%%%%%%%%%%%%%%%%%%%%%%

\subsection{Section \ref{subsection:method2} Example}
\label{section:ex3p2}
In this example, we choose,
\begin{equation}
    A = \left[ 
    \begin{array}{cc}
    0 & 1 \\
    -16 & 1
    \end{array}
    \right], \; 
    B = \left[ 
    \begin{array}{c}
    0 \\
    1
    \end{array}
    \right], \;
    X(0) = \left[ 
    \begin{array}{c}
    10 \\
    0
    \end{array}
    \right], \;
    K  = \left[ 
    \begin{array}{cc}
    -7 & \; 4
    \end{array}
    \right]
    \label{eq_sim3}
\end{equation}
Recall from Section \ref{subsection:method2} that here we set $u = 0$ when $z = 0$ and $u = -KX$ when $z = 1$. We consider three cases to compare results. In all cases we use the same $Q$ matrix. However, the sensing-cost weight $s$ is varied. The weights are:
\begin{equation}
    Q = \left[ 
    \begin{array}{cc}
    500 & 0 \\
    0 & 0
    \end{array}
    \right]\!\!, \; 
    s = \left\{ 
    \begin{array}{rl}
    0.5 & \; \!\!\text{Case 1: Low Sensing-Cost} \\
    50 & \; \!\!\text{Case 2: Interm. Sensing-Cost} \\
    500 & \; \!\!\text{Case 3: High Sensing-Cost} 
    \end{array}
    \right.
    \label{eq_sim4}
\end{equation} 

The TPBVPs for all cases of Eq.(\ref{eq_sim4}) were solved using the ``Relaxation Method''. The simulation results are presented in Fig.\ref{fig:RelaxBvp_3p2_2ndorder}. Figures \ref{fig:RelaxBvp_3p2_2ndorder}(a)-(d) compare the state and co-state trajectories. Figures \ref{fig:RelaxBvp_3p2_2ndorder}(e)-(g) plot the switching state $z$ for the three cases. From Figs.\ref{fig:RelaxBvp_3p2_2ndorder}(e)-(g), it is clear that high sensing-cost $s$ leads to a smaller net sensing interval. The net sensing durations are,
\begin{equation}
\int_0^5 z dt = \left\{ \begin{array}{ll}
    \text{Case 1:} & \;3.41\text{s} \\
    \text{Case 2:} & \;2.12\text{s} \\
    \text{Case 3:} & \;1.65\text{s}
\end{array}
\right.
\end{equation}
The effect of lesser sensing can be seen in the state variables, especially in $x_1$ in Fig.\ref{fig:RelaxBvp_3p2_2ndorder}(a). The inset in Fig.\ref{fig:RelaxBvp_3p2_2ndorder}(a) shows that $x_1$ performance is progressively degraded as $s$ increases. This is due to lesser use of the stabilizing control, $u = -KX$. Note that in this example, $A$ is not Hurwitz but $(A-BK)$ is. Figure \ref{fig:RelaxBvp_3p2_2ndorder} also confirms the correct implementation of the boundary conditions for states and co-states, as given in Eq.(\ref{eq_bc2}). Figure \ref{fig:RelaxBvp_3p2_2ndorder}(h) compares the control inputs for the three cases. The inset figure in Fig.\ref{fig:RelaxBvp_3p2_2ndorder}(h) shows progressively shorter duration of control application from Case 1 to 3.

%%%%%%%%%%%%%%%%%%%%%%%%%%%%%%%%%%%%%%%%%%%%%%%%%%%%%%%%%%%%%%%%%%%%%%%%%%%
\subsection{Section \ref{subsection:method3} Example}
\label{section:ex3p3}
In this example, we choose,
\begin{equation}
    A = \left[ 
    \begin{array}{cc}
    0 & 1 \\
    -16 & 1
    \end{array}
    \right], \quad 
    B = \left[ 
    \begin{array}{c}
    0 \\
    1
    \end{array}
    \right], \quad
    X(0) = \left[ 
    \begin{array}{c}
    10 \\
    0
    \end{array}
    \right]
    \label{eq_sim5}
\end{equation}
Recall from Section \ref{subsection:method3} that here we set $u = 0$ when $z = 0$ and $u = u^*$, Eq.(\ref{eq_uopt}), when $z = 1$. We consider three cases to compare results. In all cases we use the same $Q$ matrix. However, the weights $R$ and $s$ are varied. The weights are:
\begin{equation}
    Q = \left[ 
    \begin{array}{cc}
    5000 & 0 \\
    0 & 0
    \end{array}
    \right]\!\! , \; 
    s = \left\{
    \begin{array}{rl}
    \text{Case 1:} & 0.5 \\
    \text{Case 2:} & 50 \\
    \text{Case 3:} & 50 
    \end{array}
    \right.\!\! , \;
    R = \left\{
    \begin{array}{rl}
    \text{Case 1:} & 5 \\
    \text{Case 2:} & 5 \\
    \text{Case 3:} & 50 
    \end{array}
    \right.\!
    \label{eq_sim6}
\end{equation}

The TPBVPs for all cases of Eq.(\ref{eq_sim6}) were solved using the ``Relaxation Method". The simulation results are presented in Fig.\ref{fig:RelaxBvp_3p3_2ndorder}. Figures \ref{fig:RelaxBvp_3p3_2ndorder}(a)-(d) compare the state and co-state trajectories. Figures \ref{fig:RelaxBvp_3p3_2ndorder}(e)-(g) plot the switching state $z$ for the three cases. The net sensing durations are,
\begin{equation}
\int_0^5 z dt = \left\{ \begin{array}{ll}
    \text{Case 1:} & \;2.195\text{s} \\
    \text{Case 2:} & \;1.48\text{s} \\
    \text{Case 3:} & \;3.395\text{s}
\end{array}
\right.
\label{eq_sim7}
\end{equation}
The integral control effort for the three cases are:
\begin{equation}
\int_0^5 \!\!\!\! z\, u^2 dt = \!\! \int_0^5 \!\!\!\! z \left(-0.5R^{-1}B^T \lambda \right)^2 \!\! dt = \left\{ \begin{array}{ll}
    \!\text{Case 1:} & 1.45\times 10^6 \\
    \!\text{Case 2:} & 1.44\times 10^6 \\
    \!\text{Case 3:} & 7.77\times 10^5
\end{array}
\right.
\label{eq_sim8}
\end{equation}
Note that from Case 1 to 2, the cost of sensing was increased from $s = 0.5$ to $s = 50$. This caused the sensing duration to decrease, as evident from Eq.(\ref{eq_sim7}). Between Cases 1 and 2, note that the integral control effort remains approximately unchanged, as evident from Eq.(\ref{eq_sim8}). From Case 2 to 3, the cost of control was increased from $R = 5$ to $R = 50$. As expected, this led to a reduction in the integral control effort, as seen in Eq.(\ref{eq_sim8}). The state $x_1$ in Fig.\ref{fig:RelaxBvp_3p3_2ndorder}(a) shows an interesting behavior between the three cases. From Case 1 and 2, we observe that steady-state behavior of $x_1$ degrades. This is attributed to the higher sensing-cost $s$ in Case 2. However, from Case 2 to 3, we notice a degradation in the transients. This is due to the increase in $R$ which leads to lower control effort in Case 3. This is also reflected in the control inputs plotted in Fig.\ref{fig:RelaxBvp_3p3_2ndorder}(h). Compared to Cases 1 and 2, the control input signal of Case 3 has lower magnitudes due to a higher $R$ value. The inset figure in Fig.\ref{fig:RelaxBvp_3p3_2ndorder}(h) confirms that the control input switches to zero when $z = 0$.

%%%%%%%%%%%%%%%%%%%%%%%% COMBINED FIGURE FOR SECT 5.4 %%%%%%%%%%%%%%%%%%%%%%%%
% Note: place the figure here to keep it close to the related text
\begin{figure*}[t]
\begin{center}
\includegraphics[width = 6.0in]{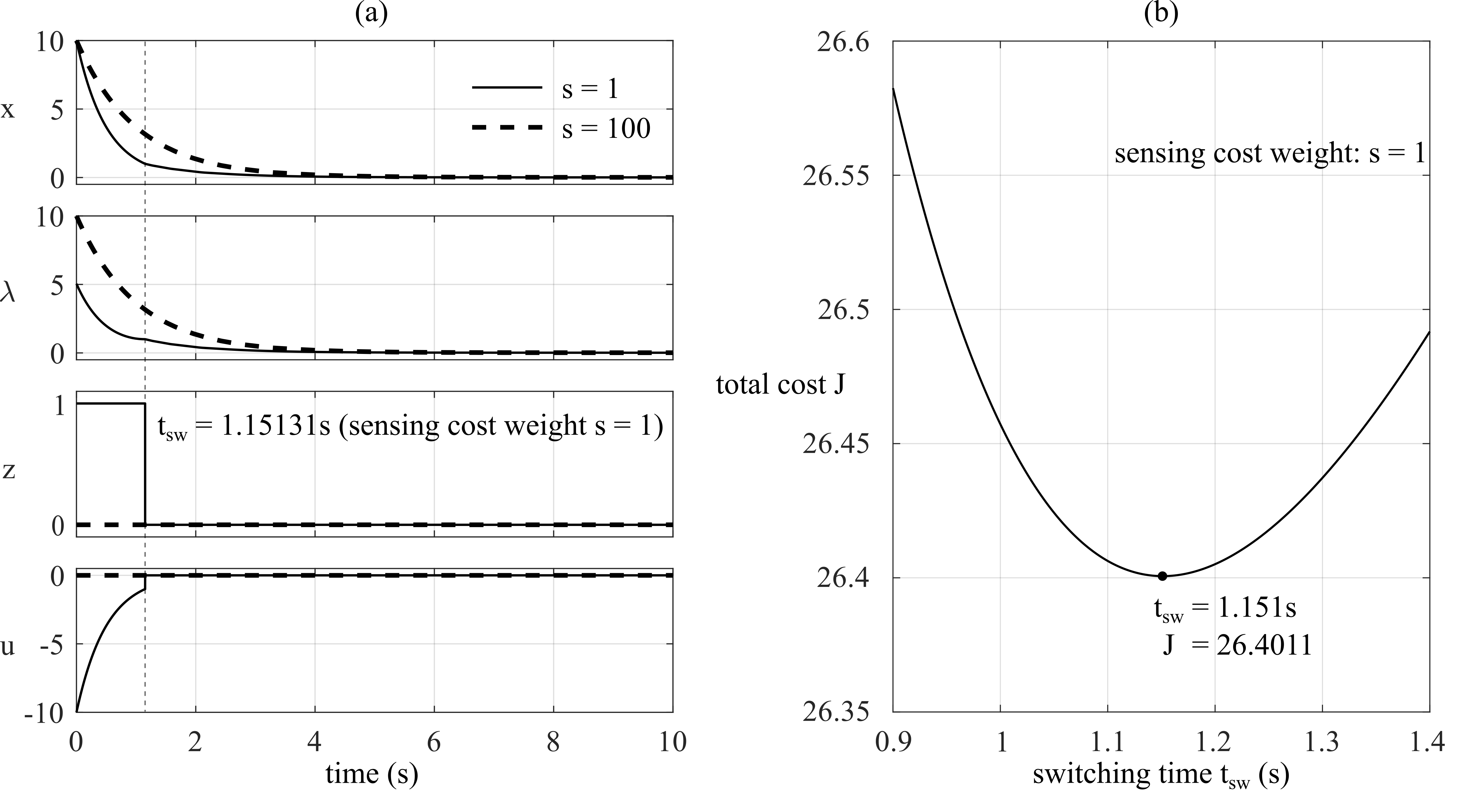}
\vspace{-.1in}
\caption{Numerical solution for Sec. \ref{sec_1st} example}
\label{fig:Sec5_4Result}
\end{center}
\end{figure*}
%%%%%%%%%%%%%%%%%%%%%%%%%%%%%%%%%%%%%%%%%%%%%%%%%%%%%%%%%%%%%%%%%%%%%%%%%%%%%%

%%%%%%%%%%%%%%%%%%%%%%%%%%%%%%%%%%%%%%%%%%%%%%%%%%%%%%%%%%%%%%%%%%%%%%%%%%%
\subsection{Section \ref{sec_1st} Example}
\label{subsection:NumericalSolution}

This section presents simulation results for a first-order system with the following parameters:
\begin{equation}
    \begin{array}{llll}
     a = -1, & \quad b = 1, & \quad k = 1, & \quad x_0 = 10,  \\
     q_1 = 1, & \quad q_2 = 0, & \quad s = 1, & \quad t_f = 10.
    \end{array}
\label{eq_mopa}
\end{equation}
The example uses a built-in function in MATLAB named \textit{bvp4c} to solve the TPBVP numerically. It is implemented using the three-stage Lobatto IIIa formula and the fourth-order forward integration (\cite{Kierzenka2001BVP}). From the analytical result of Eq.(\ref{eqn:t_sw_1order}), the optimal switching time is $t^*_{sw} = 1.15129 \ sec$, which is very close to the result shown in Fig.\ref{fig:Sec5_4Result}(a). The difference between the results is due to: (1) The function \textit{bvp4c} uses an approximate method for time-integration, and (2) The end time of the numerical simulation ($t_f = 10\ s$) is finite.

\begin{comment}
\begin{figure}[t]
\begin{center}
\includegraphics[width = 0.75\columnwidth]{Figures/sensing_cost_1st_order_bvp4c_rev_with_input.pdf}
%\vspace{-.3in}
\caption{\hl{Numerical solution for Sec.} \ref{sec_1st} \hl{example (with $s = 1$)}} 
\label{fig:bvp4c_1storder}
\end{center}
\end{figure}
\end{comment}

Another validation result is illustrated in Fig.\ref{fig:Sec5_4Result}(b), where the system dynamics is directly simulated at varying switching instants. The switching time is swept with each step of $1$ms. The total cost is calculated at each step following Eq.(\ref{eqn:cost_1storder}). The result shows that the optimal switching time is $t_{sw,opt} = 1.151 \ sec$, which matches the results mentioned above. Figure \ref{fig:Sec5_4Result}(a) also shows the optimal solution for $s=100$ obtained numerically. For the parameters chosen and the criterion for the existence of a $t_{sw}$, i.e., $s<-\left({(q_1bk+q_2a)x_0^2}/{a}\right)$, we note that $t_{sw}$ exists if $s < 100$. Accordingly, in the case of $s=100$, the optimal solution should be $z=0\ \forall \,t$. This is in agreement with the solution in Fig.\ref{fig:Sec5_4Result}(a).

\begin{comment}
\begin{figure}[t]
\begin{center}
\includegraphics[width = 0.8\columnwidth]{Figures/sensing_cost_1st_order_bvp4c_tsweep_rev.eps}
%\vspace{-.15in}
\caption{Total cost vs. switching time} 
\label{fig:bvp4c_1storder_tsweep}
\end{center}
\end{figure}

\begin{figure}[t]
\begin{center}
\includegraphics[width = 0.75\columnwidth]{Figures/sensing_cost_1st_order_bvp4c_highS_rev.eps}
%\vspace{-.2in}
\caption{Section \ref{sec_1st} example (with $s = 100$)} 
\label{fig:bvp4c_1storder_highS}
\end{center}
\end{figure}
\end{comment}
%%%%%%%%%%%%%%%%%%%%%%%%%%%%%%%%%%%%%%%%%%%%%%%%%%%%%%%%%%%%%%%%%%%%%%%%%%%%%%%%%%%%%%%%%%%%%%%
\section{A Practical Example: Waste Water Treatment Plant}
\label{section:WWTP}
Human activities are increasing the concentrations of nitrogen and phosphorus in the environment. Excessive levels of these nutrients lead to water pollution, causing adverse effects on human health, environmental degradation, and economic losses \cite{GarciaSanz2017}. Waste water treatment plants (WWTPs) play a critical role in removing nitrogen and phosphorus from water. In these systems, water quality sensors, such as those measuring ammonia and nitrate concentrations, are costly and require regular maintenance. This motivates the application of a sensing cost framework to minimize the total sensor usage time while preserving the performance of the WWTP.

In this section, we consider a linearized WWTP model around the nominal operating point defined by the ammonia and nitrate concentrations. We incorporate sensing cost into an output feedback control scheme to reduce sensor usage time while maintaining control performance near the nominal operating condition.
\subsection{State-space representation}
From \cite{GarciaSanz2017}, the MIMO transfer function matrix of a WWTP can be expressed as:
\begin{equation}
    \left[ 
    \begin{array}{c}
         Y_1(s) \\
         Y_2(s)
    \end{array}
    \right]
    =
    \left[
    \begin{array}{cc}
        p_{w11}(s) & p_{w12}(s)  \\
        p_{w21}(s) & p_{w22}(s) 
    \end{array}
    \right]
    \left[
    \begin{array}{c}
         U_1(s)\\
         U_2(s) 
    \end{array}
    \right]
\end{equation}

Where the outputs $y_1(t)$ and $y_2(t)$ are the concentration of ammonia $(NH_4)$ and nitrates $(NO_3)$ in the effluent, respectively, the input $u_1(t)$ is the air flow injected in the aerobic tank and $u_2(t)$ is the external sludge recirculation flow from settler to anaerobic tank. The transfer function $p_{wij}(s),\ i,j \in \{ 1,2 \}$ are obtained by linearizing the WWTP model about the working point, where the concentrations of ammonia and nitrates are 1 and 4 $g/m^3$, respectively.
\begin{equation}
\begin{aligned}
p_{w11}(s) &= \frac{k_{11}}{1+s/a_{11}} \\
p_{w12}(s) &= \frac{
k_{12}\left( 1 + s/z_{12,1} \right) \left( 1 + s/z_{12,2} \right)
}{
\left( 1 + s/a_{12} \right) \left[
1 + \left( 2\zeta_{12}/\omega_{n12}\right)s +
\left( s/\omega_{n12}\right)^2
\right]
}\\
p_{w21}(s) &= \frac{k_{21}}{1+s/a_{21}}, \quad
p_{w22}(s) = \frac{k_{22}}{1+s/a_{22}}    
\end{aligned}    
\end{equation}

where the parameters are:
\begin{itemize}
    \item $k_{11} = -0.04$, $a_{11} = 7.25\times10^{-5}$.
    \item $k_{12} = -6.239 \times 10^{-4}$, $z_{12,1} = 7.534 \times 10^{-4}$, $z_{12,2} = -3.17 \times 10^{-5}$, $\omega_{n12} = 4.58 \times 10^{-4}$, $\zeta_{12} = 0.8493$.
    \item $k_{21} = 0.0464$, $a_{21} = 1.008\times10^{-4}$.
    \item $k_{22} = -2 \times 10^{-5}$, $a_{22} = 1.67\times10^{-4}$.
\end{itemize}

%%%%%%%%%%%%% FIGURE: WWTP RESULT %%%%%%%%%%%%%%%%%%%%%%%%%%%%%%
\begin{figure*}
    \centering
    \includegraphics[width=0.9\linewidth]{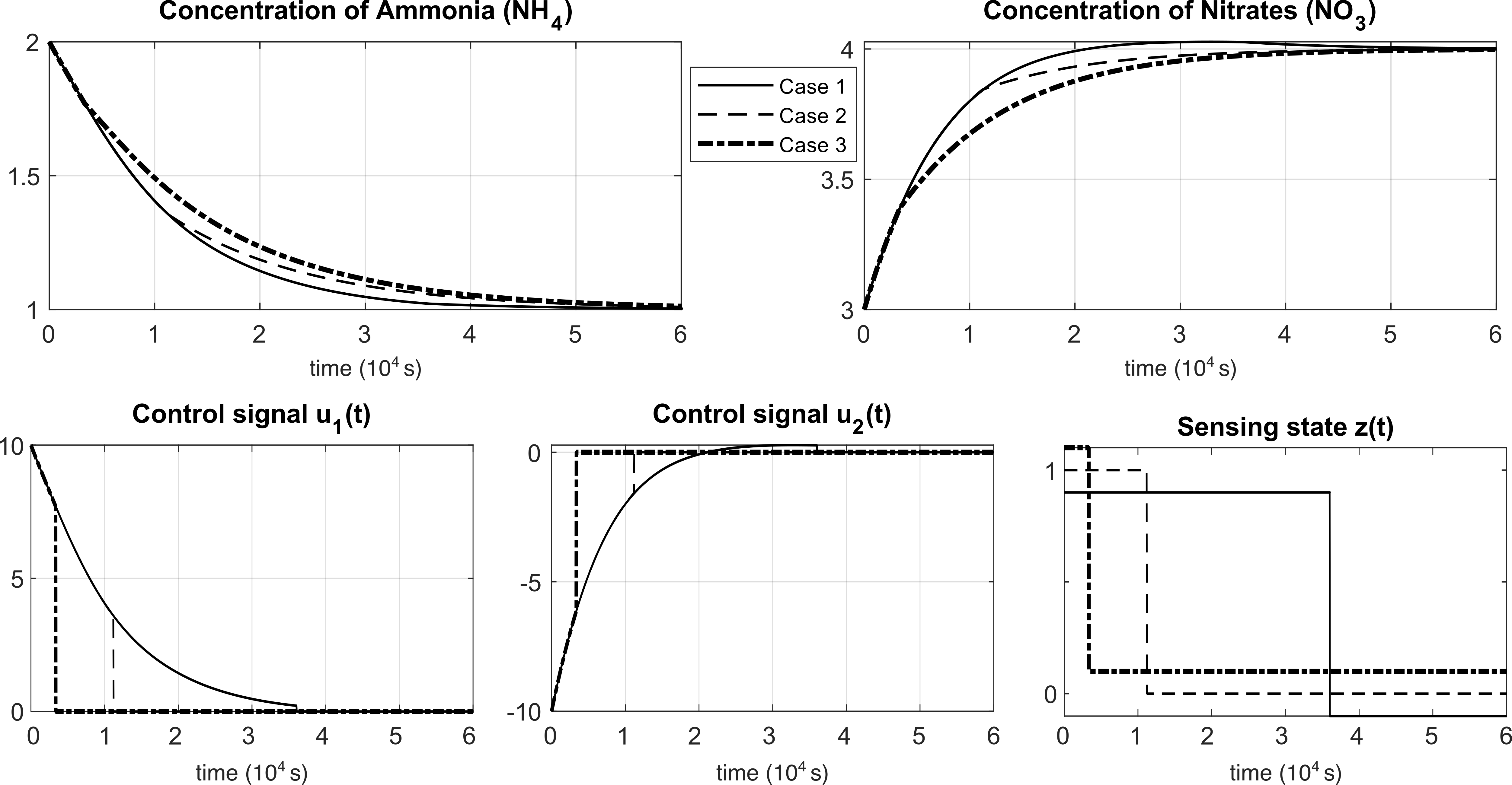}
    \caption{Simulation results for three sensing cost values: low ($s=0.001$), intermediate ($s=0.1$), and high ($s=0.5$). Concentration plots show true concentrations (not outputs $y_1$ and $y_2$), and the sensing state $z(t)$ plots are offset for clarity}
    \label{fig:wwtp_compare}
\end{figure*}
%%%%%%%%%%%%% FIGURE: WWTP RESULT %%%%%%%%%%%%%%%%%%%%%%%%%%%%%%

Because there are three first-order and one third-order transfer functions, a state-space representation with order 6 can be employed to describe this system:
\begin{equation}
    \left\{ 
    \begin{aligned}
        \dot X &= AX + BU \\
        Y &= CX + DU
    \end{aligned}
    \right.
\end{equation}

where:
\begin{subequations}
    \begin{equation*}
        A = \left[ 
        \begin{array}{cccccc}
             -7.25\times10^{-5} & \boldsymbol{0}_{1\times3} & 0 & 0 \\
             \boldsymbol{0}_{3\times1} & A_{12} & \boldsymbol{0}_{3\times1} & \boldsymbol{0}_{3\times1} \\
             0 & \boldsymbol{0}_{1\times3} & -1.008\times10^{-4} & 0  \\
             0 & \boldsymbol{0}_{1\times3} & 0 & -1.67\times10^{-4} 
        \end{array}
        \right]
    \end{equation*}
    \begin{equation*}
        A_{12} = \left[ 
        \begin{array}{ccc}
             0 & 1 & 0  \\
             0 & 0 & 1  \\
             -1.687\times10^{-11} & -2.723\times10^{-7} & -8.584\times10^{-4} \\
        \end{array}
        \right]
    \end{equation*}
    \begin{equation*}
        B = \left[ 
        \begin{array}{cc}
             -2.9\times10^{-6} & 0 \\
             0 & 0 \\
             0 & 0 \\
             0 & 1 \\
             4.677\times10^{-6} & 0 \\
             0 & -3.34\times10^{-9}
        \end{array}
        \right]
    \end{equation*}
    \begin{equation*}
        C = \left[ 
        \begin{array}{cccccc}
             1 & -1.052\times10^{-16} & 3.18\times10^{-12} & 4.406\times10^{-9} & 0 & 0  \\
             0 & 0 & 0 & 0 & 1 & 1
        \end{array}
        \right]
    \end{equation*}
    \begin{equation*}
        D = \boldsymbol{0}_{2\times2}
    \end{equation*}
\end{subequations}

We assume that only the outputs can be measured, therefore, we propose an output-feedback control:
\begin{equation}
    U = -K_y Y = \left[
    \begin{array}{cc}
         K_{y11} & K_{y12}  \\
         K_{y21} & K_{y22} 
    \end{array}
    \right]
    \left[
    \begin{array}{c}
        y_1\\
        y_2
    \end{array}
    \right]
\end{equation}

To determine the gain matrix $K_y$, we first apply relative gain array (RGA) analysis to assess the input--output pairings. The result indicates that, in the low-frequency range, the diagonal pairing (use $u_1$ to control $y_1$ and $u_2$ to control $y_2$) is appropriate for controlling the MIMO system. Based on trial-and-error tuning, the gain matrix is selected as $K_y = -10\boldsymbol{I}_2$. Denote $z \in \left\{ 0, 1 \right\}$ as the sensing state, we will find the optimal sequence $z^*(t)$ to optimize the following cost functional:
\begin{equation}
    J = \int_{t_0}^{t_f}{\left(Y^TQ_yY + sz\right)dt }
    = \int_{t_0}^{t_f}{(X^T \underbrace{C^TQ_yC}_{Q}X + sz)dt }
\end{equation}

By applying the results presented in section \ref{subsection:method2}, we have the co-state differential equation:
\begin{equation}
\begin{aligned}
    \dot\lambda &= 
    %- \frac{\partial H}{\partial X} = 
    - 2C^TQ_yCX - (A - zBK_yC)^T\lambda \\
    &= - 2C^TQ_yY - (A - zBK_yC)^T\lambda
\end{aligned}
\end{equation}

and the switching condition:
\begin{equation}
    z = \left\{ 
    \begin{array}{ccc}
        1 & \text{ when } & s < \lambda^{*T}BK_yY^* \\
        0 & \text{ when } & s > \lambda^{*T}BK_yY^*
    \end{array}
    \right.
\end{equation}

\subsection{Simulation Results}
\begin{comment}
\begin{equation*}
    X_0 = 
    \begin{array}{ccc}
         \left[9.634\times10^{-1}\right. & -3.483\times10^{14} & 1.665\times10^{2} \text{ ...} \\
         2.96\times10^{-1} & -1.117 & \left.1.175\times10^{-1}\right]
    \end{array},
\end{equation*}    
\end{comment}
We use the MATLAB built-in function \textit{bvp4c} to simulate the WWTP model in a time span of $60000$ seconds. The initial condition of the state vector is set at $X_0 = \left[9.634\times10^{-1}\right.$, $-3.483\times10^{14}$, $1.665\times10^{2}$, $2.96\times10^{-1}$, $-1.117$, $\left.1.175\times10^{-1}\right]$, which corresponds to a starting point with the output of $Y(0)~=~\left[1,\ -1
\right]^T$. The chosen output cost matrix is $
Q_y = \boldsymbol{I}_2$. The simulation is conducted for three different values of the sensing cost coefficient $s$:
\begin{equation*}
    s = \left\{ 
    \begin{array}{cl}
        0.001 & \text{Case 1: Low sensing cost} \\
        0.1 & \text{Case 2: Intermediate sensing cost}\\
        0.5 & \text{Case 3: High sensing cost}
    \end{array}
    \right.
\end{equation*}

Figure~\ref{fig:wwtp_compare} and Table~\ref{tab:wwtp_sensing_duration} illustrate the comparison among three scenarios. When the sensing cost is incorporated into the model, a trade-off between output cost and sensing cost can be observed. As the sensing cost increases (i.e., sensing becomes more expensive), the optimal total sensing duration decreases to minimize the overall cost. In addition, as the sensing cost coefficient increases, the output performance degrades, as shown in Fig.~\ref{fig:wwtp_compare}. Specifically, a higher sensing cost leads to a longer settling time, as well as a higher output cost, as indicated in Table~\ref{tab:wwtp_sensing_duration}.

\begin{table}[ht]
    \centering
    \caption{Performance comparison between three cases}
    \begin{tabular}{c|c|c}
         Case & Output cost & Net sensing duration (s) \\ \hline
         1 & 8984 & 36033 \\
         2 & 9288 & 11209 \\
         3 & 11035 & 3364
    \end{tabular}
    \label{tab:wwtp_sensing_duration}
\end{table}

%%%%%%%%%%%%%%%%%%%%%%%%%%%%%%%%%%%%%%%%%%%%%%%%%%%%%%%%%%%%%%%%%%%%%%%%%%%%%%%%%%%%%%%%%%%%%%%
\section{A Shrinking Horizon Implementation}
\label{sec_sh}

%%%%%%%%%%%%%%%%%%%%%%%% Shrinking horizon scheme %%%%%%%%%%%%%%%%%%%%%%%%
\begin{figure}[b]
\begin{center}
\includegraphics[width = 0.95\columnwidth]{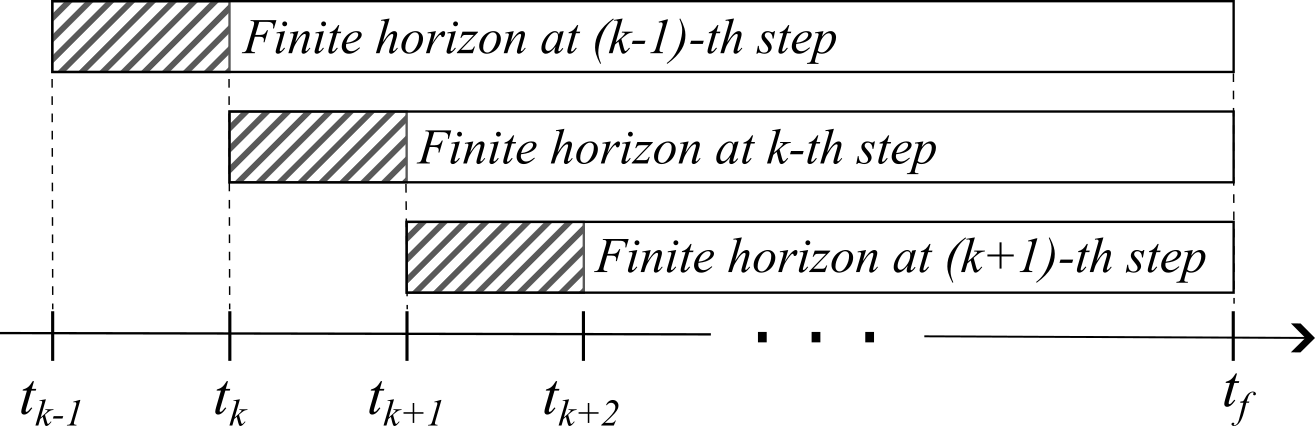}
%\vspace{-.2in}
\caption{Shrinking horizon scheme (adapted from \cite{Capannolo2023})} 
\label{fig:SH_diagram}
\end{center}
\end{figure}

All of the solutions presented in Section \ref{sect:theory} and examples in Section \ref{section:Examples} are considered \textit{offline} solutions, meaning that the optimal sequence of $z$ is pre-calculated based on a known, precise system model and initial conditions. Then, using the pre-calculated switching times, sensing can be switched on or off and inputs accordingly applied. In practice however, a system is always subject to uncertainties, which leads to sub-optimal performance when the offline solution is implemented. This issue is compounded for the problem addressed in Section \ref{subsection:method3}, where the optimal input of Eq.(\ref{eq_uopt}) is not directly a state feedback control. As a result, even when $z=1$, it is not explicit from Eq.(\ref{eq_uopt}) how sensing is helpful. A {\it Shrinking Horizon} (SH) implementation is a practical way to address this issue (\cite{Capannolo2023} and \cite{Skaf2010}). 

In the {\it Shrinking Horizon} method, the time span $\left[t_0, t_f\right]$ is evenly divided into multiple steps. Initially, we assume that the state $X(t_0) = X_0$ is known and we solve the optimal problem by using the numerical methods presented in Section \ref{section:Examples}. If the resulting optimal sensing state at the initial time is $z=1$, we apply the feedback control and the next initial state $X(t_1)$ is known (i.e., sensed). Conversely, if $z=0$, we apply the appropriate control input (i.e. $u = [0]$ or $\dot{u} = [0]$), but do not sense  $X(t_1)$. However, since the system model is given (albeit, may not be accurate), we can integrate the system model to get the calculated state $X_{cal}(t_1)$. We choose this state as the new initial condition and solve the optimal control problem (i.e., the TPBVP) for a time span $\left[t_1, t_f\right]$, which is over a shrinking time horizon. We repeat this process until we reach the final time step $k=N$ (Fig. \ref{fig:SH_diagram}).

We first demonstrate the Shrinking Horizon (SH) implementation for the first-order system introduced in Section~\ref{subsection:NumericalSolution} and Eq.(\ref{eq_mopa}). This allows direct comparison with the analytical solution. The SH is implemented in a time spans of $10$ sec, discretized into $1000$ steps with a sampling time of $0.01$ sec. When the model parameters are perfectly known, it is straightforward that the SH implementation produces results identical to the single-step TPBVP solution.

%%%%%%%%%%%%%%%%%%%%%%%% FIGURE: 1ST ORDER SH %%%%%%%%%%%%%%%%%%%%%%%%
% Note: place the figure here to keep it close to the related text
\begin{figure}[b]
\begin{center}
\includegraphics[width = 0.95\columnwidth]{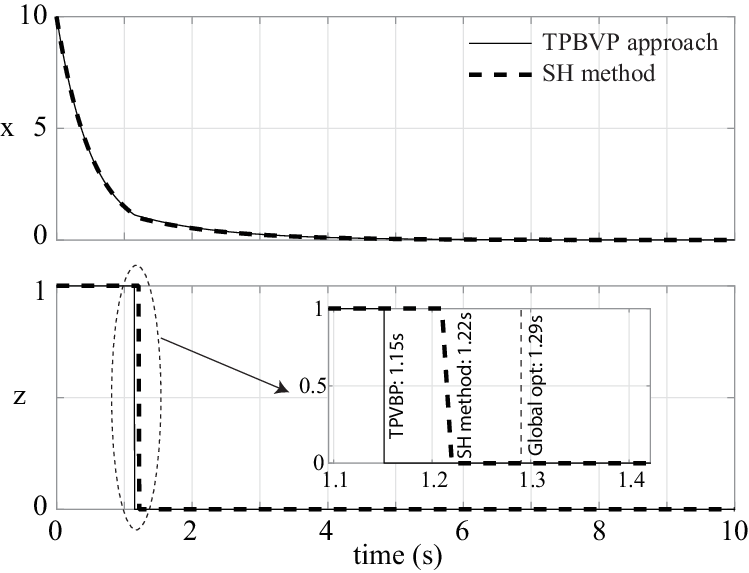}
%\vspace{-.1in}
\caption{Comparison between TBPVP and Shrinking Horizon implementation, first-order system (with model uncertainties $\delta a = 0.2$, $\delta b = 0.1$)}
\label{fig:SH_result_2}
\end{center}
\end{figure}
%%%%%%%%%%%%%%%%%%%%%%%% FIGURE: 2ND ORDER SH %%%%%%%%%%%%%%%%%%%%%%%%
% Note: place the figure here to keep it close to the related text
\begin{figure*}[t]
\begin{center}
\includegraphics[width = 1.8\columnwidth]{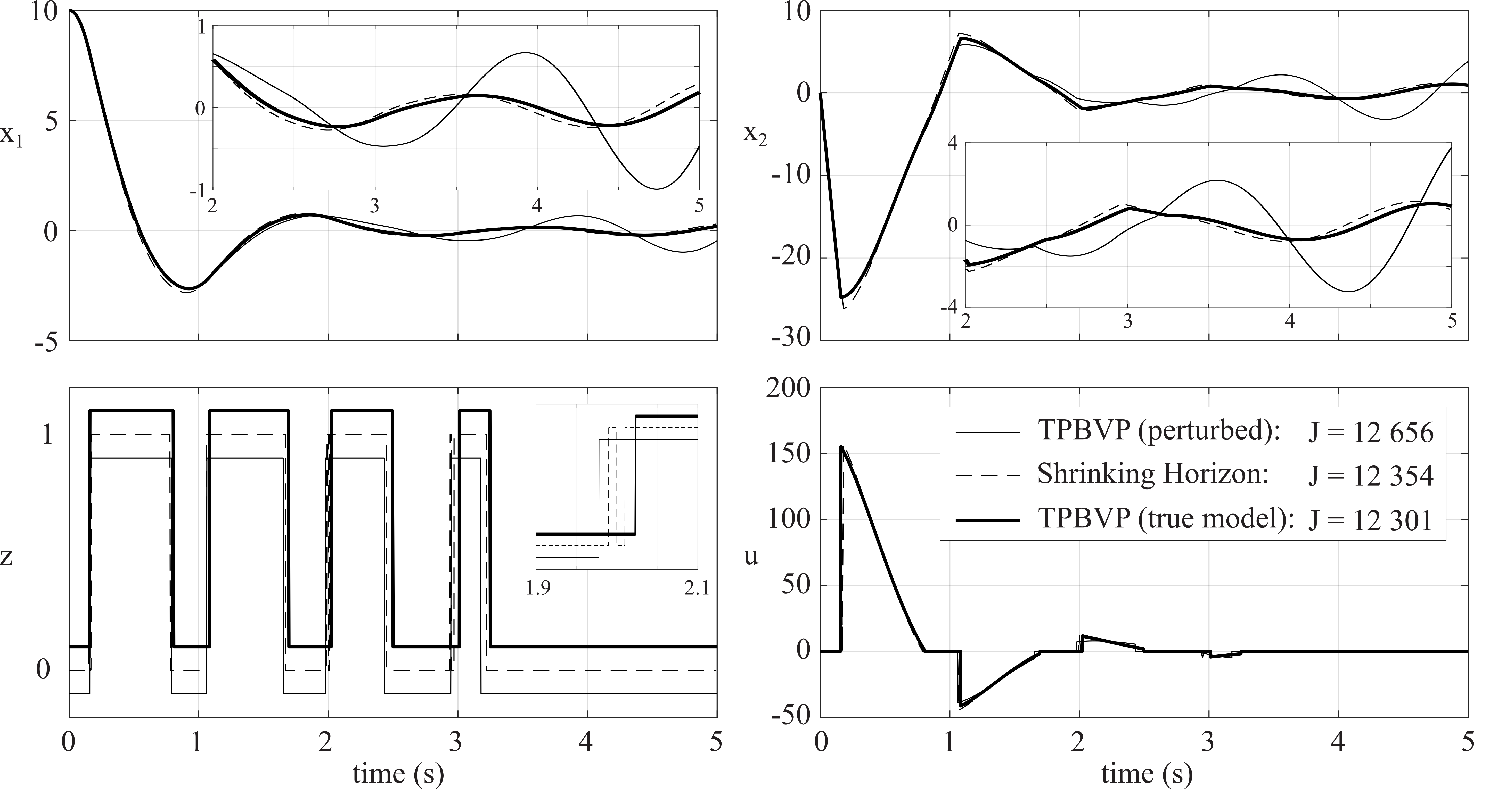}
\vspace{-.1in}
\caption{Comparison between TBPVP and Shrinking Horizon, second-order system with model uncertainty. Note that the z-plots have different offsets for easier distinction}
\label{fig:SH_result_3_2ndorder}
\end{center}
\end{figure*}
%%%%%%%%%%%%%%%%%%%%%%%%%%%%%%%%%%%%%%%%%%%%%%%%%%%%%%%%%%%%%%%%%%%%%%%%%%%%%

Next, we introduce uncertainties in the system. We assume that the \textit{true model} can be formulated as,
\begin{equation}
    \dot{x}_{true}(t) = (a+\delta a)x_{true} + (b+\delta b)u
\end{equation}
In this simulation, the uncertainty parameters are constant bias set to $\delta a = 0.2$, $\delta b = 0.1$. Applying Eq.(\ref{eqn:t_sw_1order}), the optimal switching time is calculated to be $t_{sw} = 1.29$ seconds. As shown in Fig.\ref{fig:SH_result_2}, the switching time for SH implementation and TPVBP methods are at 1.22s and 1.15s, respectively. Since the switching time derived from the SH method is consistent with the global optimum as well as the TPBVP solution, this indicates that the method is suitable for handling practical problems involving switching cost impositions, including ones that involve uncertainties.

%%%%%%%%%%%%%%%%%%% Removed
\begin{comment}
\begin{figure}[t]
\begin{center}
\includegraphics[width = 0.8\columnwidth]{Figures/TPBVP_SH_1st_order_stable_rev.eps}
\vspace{0in}
\caption{Comparison between TBPVP and Shrinking Horizon method, first-order system (no uncertainty)} 
\label{fig:SH_result_1}
\end{center}
\end{figure}
\end{comment}

%%%%%%%%%%%%%%%%%%%%%%%% 2ND ORDER SHRINKING HORIZON %%%%%%%%%%%%%%%%%%%%%%%%

Next, consider the unstable second-order system introduced in Section~\ref{section:ex3p2}, given by Eqs.(\ref{eq_sim3})--(\ref{eq_sim4}) (Case 2: $s = 50$). To demonstrate the adaptivity of the Shrinking Horizon method, the uncertainties of about 3\% to 5\% are introduced into the system matrices $A$ and $B$ as follows:

\begin{equation}
\begin{aligned}
A_{true} &= A + \delta A \\ 
&=  
\left[ \begin{array}{cc}
    0 & 1 \\
    -16 & 1
\end{array} \right]
+
\left[ \begin{array}{cc}
    0 & 0 \\
    0.7 & -0.03
\end{array} \right]
=
\left[ \begin{array}{cc}
    0 & 1 \\
    -15.3 & 0.97
\end{array} \right]
\end{aligned}
\end{equation}

\begin{equation}
B_{true} = B + \delta B =
\left[ \begin{array}{cc}
    0 \\
    1
\end{array} \right]
+
\left[ \begin{array}{cc}
    0 \\
    -0.05
\end{array} \right]
=
\left[ \begin{array}{cc}
    0 \\
    0.95
\end{array} \right]
\end{equation}

In this example, the 5-second time span is divided into $500$ steps, each with a sampling time of $0.01$ sec. Figure \ref{fig:SH_result_3_2ndorder} compares the performance of the single TPBVP approach, the SH implementation, and the TPBVP solution of the true model (used as a reference). The total cost $J$ for each method is as follows: the single TPBVP approach yields $J = 12656$, the SH implementation results in $J = 12354$, and the TPBVP for true model achieves the lowest cost with $J = 12301$. The SH method reduces the cost by 302 units relative to the single TPBVP, recovering over $85\%$ of the performance loss caused by model uncertainty. The state trajectories $x_1(t)$ and $x_2(t)$ show that the SH method effectively compensates for errors introduced by model uncertainty. While the single TPBVP exhibits unstable behavior unsuitable for long-term operation, the SH implementation maintains stability of the state vector $[x_1,x_2]^T$. This performance gain demonstrates the SH method’s ability to better track the optimal trajectory in uncertain environments. 
However, the SH result may exhibit chattering (high-frequency oscillations), as observed at $t = 2\;s$. This occurs when the sensing state $z$ switches from 0 to 1 (or vice versa), potentially causing a mismatch between the guessed state (propagated using the nominal matrices $A$ and $B$) and the sensed state (according to the true model $\dot{X}_{true} = A_{true}X_{true} + B_{true}u$). This behavior should be further investigated to ensure the robustness of the control scheme.
As in the first-order case, this simple example illustrates the feasibility of the SH approach and hints at its potential for more complex systems with switching costs.

While the cases presented here are basic, further investigation across a wider range of systems is needed to fully assess its viability. Nevertheless, these results confirm the basic feasibility of the method and its suitability for the formulation in Sec.~\ref{subsection:method3}, where, whenever $z = 1$, the SH implementation uses the measured $X(t)$ to solve the TPBVP and obtain $\lambda(t)$ for use in Eq.~(\ref{eq_uopt}).

%%%%%%%%%%%%%%%%%%%%%%%%%%%%%%%%%%%%%%%%%%%%%%%%%%%%%%%%%%%%%%%%%%%%%%%%%
\section{Conclusion}
\label{sec_conclu}
%This study presents a novel control design framework for LTI systems, using a binary switching variable, z, to toggle state sensing between active (z=1) and inactive (z=0). For each scenario, distinct control input formulations are examined: during active sensing, we consider predefined stabilizing control feedback and an optimal control formulation, while for inactive sensing, we explore nullifying the input ($u=[0]$) and freezing the input $\dot{u} = [0]$. The switching sequence of z is optimized to minimize a cost functional incorporating penalties on system states, control effort, and sensing duration. This optimization is achieved using Pontryagin’s Maximum Principle to minimize the associated Hamiltonian. 

This study presents a novel optimal control framework for LTI systems that incorporates cost of sensing in its performance index. In this introductory work, a binary switching variable, $z$, toggles state-sensing between active ($z=1$) and inactive ($z=0$). For each scenario, distinct control input formulations are examined. During active sensing, we consider a predefined state-feedback control and an optimal control. During inactive sensing, we explore zero input ($u=[0]$) and constant input $\dot{u} = [0]$. The switching sequence of $z$ is optimized to minimize a cost functional incorporating penalties on system states, control effort, and sensing duration. This optimization is achieved using the Pontryagin’s Minimum Principle.

Various numerical simulations were conducted to validate the proposed framework. In Section \ref{section:ex3p1}, where a stabilizing feedback, $u = -KX$ is employed during active sensing and input is held constant during inactive sensing, very high penalties on sensing-cost resulted in divergent trajectories due to insufficient sensing duration. This observation underscores the importance of carefully selecting the weight of sensing-cost. Section \ref{section:ex3p3} demonstrates that a joint optimization can effectively balance sensing-costs, control efforts, and state stabiliÍzation, ensuring a balanced trade-off between system performance and resource utilization. A reduced-form equation for the infinite-horizon multi-dimensional case with single switching point is derived, and a closed-form solution is obtained for the infinite-horizon first-order case. An example for the first-order system is also provided to validate the derived closed-form solution, and a practical waste water treatment plant example is presented to demonstrate the applicability of the framework to systems with costly, maintenance-intensive sensors.

The resulting TPBVPs are initially solved offline, with optimal sensing schedules pre-computed assuming accurate system models. However, this offline approach may be impractical and yield suboptimal results in the presence of uncertainties. To address this, a Shrinking Horizon approach is taken, solving the optimal control problem at each time step while updating the initial state based on the sensing status. Simulations show that this method achieves switching times comparable to the offline TPBVP solutions, highlighting its potential for practical applications. Finally, we note that the proposed framework is formulated over a fixed finite time horizon and focuses on cost-based optimality rather than asymptotic stabilization. As a result, open-loop unstable systems can be considered over the horizon of interest.
While the current work focuses on LTI systems, the theoretical framework provides a strong foundation for extending the proposed approach to more complex systems. Additional formulations of sensing cost, such as selective sensing, switching between state and output feedback (e.g., via observers), and open-loop estimation, are possible extensions.

\bmsubsection*{Funding}
The authors declare that no external funding was received for this work.

\bmsubsection*{Conflicts of Interest}
The authors declare that they have no known competing financial interests or personal relationships that could have appeared to influence the work reported in this paper.

\bmsubsection*{Data availability statement}
The data that support the findings of this study are available from the corresponding author upon reasonable request.

%\nocite{*}% Show all bib entries - both cited and uncited; comment this line to view only cited bib entries;

%\bibliographystyle{wileyNJD-Chicago-lastoo}

%\bibliography{sensing-cost.bib} %% <=== change this to name of your bib file

\end{document}